\documentclass[12pt,a4,authoryear]{article}
\usepackage[authoryear]{natbib}
\usepackage{authblk}
\usepackage{amsmath}
\usepackage{amsthm}
\usepackage{amssymb}
\usepackage{xcolor}
\usepackage{CJK}
\usepackage{blkarray}
\usepackage{geometry}
\usepackage{mathrsfs}
\usepackage{longtable}
\usepackage{booktabs}
\usepackage{epstopdf}
\usepackage{diagbox}
\usepackage{amsfonts}
\usepackage[colorlinks,
linkcolor=blue,
anchorcolor=blue,
citecolor=blue
]{hyperref}
\usepackage{multicol}
\usepackage{multirow}
\usepackage{algorithm}
\usepackage{algorithmic}
\usepackage{rotating}
\usepackage{caption}
\usepackage{subcaption}
\usepackage{setspace}
\usepackage{graphicx}
\usepackage[toc,page,title,titletoc,header]{appendix}
\usepackage{indentfirst} 
\def\bSig\mathbf{\Sigma}

\newtheorem{lemma}{Lemma}
\newtheorem{rmk}{Remark}

\newcommand{\dT}{\mathsf{T}}
\allowdisplaybreaks[4]

\def\squarebox#1{\hbox to #1{\hfill\vbox to #1{\vfill}}}
\def\boxit#1{\vbox{\hrule\hbox{\vrule\kern6pt
			\vbox{\kern6pt#1\kern6pt}\kern6pt\vrule}\hrule}}

\begin{document}
	\title{Summary Statistics Knockoffs Inference with Family-wise Error Rate Control}
	\author[1,$*$]{Catherine Xinrui Yu}
	\author[2,$*$]{Jiaqi Gu}
	\author[3,$*$]{Zhaomeng Chen}
	\author[2,4,$\dagger$]{Zihuai He}
	\affil[1]{Department of Statistics, The Chinese University of Hong Kong}
	\affil[2]{Department of Neurology and Neurological Sciences, Stanford University}
	\affil[3]{Department of Statistics, Stanford University}
	\affil[4]{Department of Medicine (Biomedical Informatics Research), Stanford University}
	
	\affil[$*$]{Equal contribution.}
	\affil[$\dagger$]{Corresponding author: Zihuai He, zihuai@stanford.edu}
	
	\date{}                     
	\setcounter{Maxaffil}{0}
	\renewcommand\Affilfont{\itshape\small}
	\maketitle

	\begin{abstract}
		Testing multiple hypotheses of conditional independence with provable error rate control is a fundamental problem with various applications. To infer conditional independence with family-wise error rate (FWER) control when only summary statistics of marginal dependence are accessible, we adopt GhostKnockoff to directly generate knockoff copies of summary statistics and propose a new filter to select features conditionally dependent to the response with provable FWER control. In addition, we develop a computationally efficient algorithm to greatly reduce the computational cost of knockoff copies generation without sacrificing power and FWER control. 
 Experiments on simulated data and a real dataset of Alzheimer's disease genetics demonstrate the advantage of proposed method over the existing alternatives in both statistical power and computational efficiency.
		
	\end{abstract}

	\maketitle
	
	\section{Introduction}
	\label{s:intro}	
	
	Inference of conditional independence is a fundamental topic in statistics and machine learning and has received lots of interests in various fields. For example, genetic researchers are interested in identifying variants that are associated with diseases conditioning on other variants for genomic medicine development \citep{khera2017genetics,zhu2018causal} and financial researchers utilize conditional dependence between different indexes and factors to explain economic phenomena \citep{patel2015predicting}. As pointed out by \citet{dawid}, it is important to develop statistical methods for conditional dependence inference, either Bayesian or frequentist. Beginning with the widely-used likelihood ratio test and the Pearson's $\chi^2$ test, last several decades have witnessed various parametric or nonparametric approaches. For example, \citet{daudin1980partial} proposes a parametric test based on the linear model  while \citet{JMLR:v15:peters14a} develop a semiparametric test under additive models. 
 In addition, the use of resampling procedures (including permutation and bootstrap) is also explored to boost computational efficiency of nonparametric tests of conditional independence \citep{doran2014permutation, sen2017model}.

	As the aforementioned tests are proposed to infer a single hypothesis, directly applying them to simultaneously infer multiple hypotheses of conditional independence without any adjustment would result in type-I error rate inflation. To address the multiplicity issue, various correction approaches have been proposed, among which the earliest is the widely used Bonferroni correction \citep{dunn1961multiple}. Rejecting nulls whose $p$-values are smaller than the target family-wise error rate (FWER) divided by the number of nulls, the Bonferroni correction would easily suffer power loss as its empirical FWER is far less than the target level when $p$-values are correlated.  Several improvements of the Bonferroni correction have appeared in the literature, including the Šidák correction \citep{Sidak}, Holm's step-down procedure \citep{holm1979simple} and Hochberg's step-up procedure \citep{Hochberg1988}. Inspired by exploratory studies where false discoveries are acceptable under a specified proportion, \citet{benjamini1995controlling} propose the false discovery rate (FDR) as an alternative of FWER and develop a step-up inference procedure with FDR control. Other multiple testing procedures with FDR control include works of \citet{Yben} and \citet{Storey2003}.

	However, the above adjustments of multiple testing usually suffer different limitations in practice. Taking $p$-values as input, these adjustments can not be used in high-dimensional regression models where $p$-values may not exist or are hard to obtain. Even under the low-dimensional case, the validity of $p$-values can also be questionable due to possible model misspecification. In addition, some procedures rely on particular assumptions of dependency structure among $p$-values, which are usually hard to verify. To address these issues, a new powerful method named the knockoff filter \citep{baber,candes2018panning} has been developed recently. By synthesizing knockoff copies that mimic dependency of original features, the knockoff filter simultaneously infer conditional independence between a large set of features and a response without the need to compute $p$-values. With no model assumption on the conditional distribution of the response, such a procedure is valid even when a misspecified model is fitted \citep{candes2018panning}. 
	
	Inspired by this idea, a series of knockoff-based methods have been developed, including but not limited to multiple knockoffs \citep{james} and derandomized knockoffs \citep{zhimei} to improve the stability of inference, the $k$-FWER oriented procedure \citep{janson} to control $k$-FWER, the joint inference procedure for conditional independence of both features and feature groups \citep{katsevich2019multilayer} and the computationally efficient GhostKnockoff that only requires summary statistics of large-scale datasets \citep{he2022ghostknockoff}. 
	There also exist works on improving either power \citep{luo2022improving} or interpretability \citep{fan2020ipad} of knockoff-based procedures with FDR control. However, in many studies, FWER control remains the primary target, especially in the confirmatory phase of large-scale genetic studies and clinical trials. Among all existing knockoff-based methods, the procedure of \citet{janson} and derandomized knockoffs \citep{zhimei} are the two that provide provable FWER control. However, both methods suffer limitations in practice, where people commonly target to control FWER under $0.01$, $0.05$ or $0.1$. As a general procedure to control $k$-FWER (the probability of making at least $k$ false discoveries, which degenerates to FWER when $k=1$) for arbitrary integer $k$, directly applying the procedure of \citet{janson} with $k=1$ on FWER-controlled inference would greatly lose power when the target FWER level is small. Specifically, to control FWER under $\alpha=0.01, 0.05$ or $0.1$ as commonly in practice, the procedure of \citet{janson} only makes rejections with a small probability $2\alpha$, leading to a power bound $2\alpha$ and uninformative conclusions in most of cases. Although derandomized knockoffs \citep{zhimei}, which perform inference by running the procedure of \citet{janson} for multiple times, can also control FWER, experiments in \citet{zhimei} show that its empirical FWER is far less than the target level, suggesting the potential of power improvement.  In addition, both methods require access to individual observations in generating knockoffs, which are usually infeasible in large-scale genetic data analysis due to privacy issue and limited computational resources. Under such cases, only summary statistics that do not contain individual identifiable information are provided, including but not limited to (estimated) moments  of genotypes (e.g., mean, variance-covariance matrix, skewness and kurtosis) and $Z$-scores of Pearson's correlations between genotypes and phenotypes or disease-associate traits of interest. Thus, it is needed to develop a knockoff-based procedure to perform inference on only summary statistics with improved power and tightly controlled FWER.

	In this article, we develop a new knockoff filter to select features that are conditionally dependent on the response with provable FWER control and only access to summary statistics. By adopting the idea of GhostKnockoff \citep{he2022ghostknockoff}, our method directly generate knockoff copies of summary statistics of marginal dependence without using any individual observations. With the flexibility in paring with any knockoff models, our proposed approach generates multiple knockoffs \citep{james} instead of using the repeated strategy in derandomized knockoffs \citep{zhimei}. To circumvent the time-consuming Cholesky decomposition of a high-dimensional matrix in generating multiple knockoffs, we propose a computationally efficient algorithm  that greatly reduces computational cost of knockoffs generation without sacrificing any power and FWER control. Through extensive experiments of simulated and real data, we show that our proposed method manages to provide tight FWER control at the common level $\alpha=0.05$ and exhibits great power in comparison with existing methods.

	The rest of this article is organized as follows. In Section  \ref{section2}, we formulate multiple testing of conditional independence and introduce our new knockoff filter with theoretical studies on its FWER control. Section \ref{section3} provides strategies to improve the computational efficiency and comparisons with existing knockoff-based method in computational cost. Via extensive simulation studies in Sections  \ref{section4}, we investigate empirical performance of the proposed knockoff FWER filter in both FWER control and power  with comparisons to existing FWER-oriented knockoffs methods. We also apply our method on a  genetic data to investigate important genetic variants of Alzheimer's disease in Section \ref{section5}. Section \ref{section6} concludes with discussions.

	\section{Methodology}\label{section2}
	
	\subsection{Problem Statement}
	Consider independently and identically distributed (i.i.d.) observations $\{(\textbf{x}_i,y_i)\in\mathbb{R}^p\times \mathbb{{R}}|i=1,\ldots,n\}$ from a joint distribution $f(\mathbf{ X},Y )$ with $\textbf{X}=(X_1,\ldots,X_p)^\dT$, our interest is to test hypotheses 
	\begin{equation}\label{indep}
		H_j:	Y \perp X_j |\textbf{X}_{-j},\quad j=1,\ldots,p,
	\end{equation}
	under the joint distribution $f(\mathbf{ X},Y )$, where $\textbf{X}_{-j}  =(X_1,\ldots,X_{j-1},X_{j+1},\ldots,X_p)^\dT$.
	Believing that the conditional distribution $Y|\mathbf{ X}$ only depends upon a small subset of relevant features in $\textbf{X}$, our target is to classify each feature as relevant or not. Here, feature $X_j$ is said to be nonnull (or relevant to $Y$) if the corresponding assumption $H_j$ is not true, and null (or irrelevant to $Y$) otherwise.
	Let ${\mathcal{H}_0}\subset \{1,\dots,p\}$ denote the set of null features whose $H_j$'s are true and the remaining as $\mathcal{H}_1$. 
	In this paper, we denote $\mathcal{R}$ as the set of rejected hypotheses and $V$ as the number of false discoveries (true $H_j$'s being rejected).
	Our target is to obtain an estimate $\mathcal{R}$ of $\mathcal{H}_1$ such that
	\begin{equation}
		\text{FWER}= \mathbb{P}(V\geq 1), \quad \text{where }V=\#\{j:j\in \mathcal{H}_0\cap \mathcal{R}\},
	\end{equation}
	is not larger than $\alpha$ ($0<\alpha<1$).

	\subsection{GhostKnockoff}\label{GK}

In large-scale genome-wide association studies, individual data required for  knockoffs generation are generally inaccessible. Even available, computational cost can be extremely large due to the large sample size, which blocks the potential of existing individual knockoff-based methods in analyzing large-scale data. Recently, \citet{he2022ghostknockoff} developed GhostKnockoff, which only requires $Z$-scores of Pearson's correlations between features and the response and the (estimated) correlation matrix $\boldsymbol{\Sigma}$ of features  to directly generate knockoff copies of  $Z$-scores as follows.


	Let $\mathbf{x}_i=(x_{i1},\dots,x_{ip})^\dT$ and $y_i$ be the feature vector and the outcome of the $i$-th individual ($i=1,\ldots,n$) under the distribution $f(\mathbf{X},Y)$. Without loss of generality, we assume that both features and response are standardized with mean 0 and variance 1.
	To measure importance of different features, we consider $Z$-scores of Pearson's correlations,
	\begin{equation}\label{zscore}
		\mathbf{Z}=(Z_1,\ldots,Z_p)^\dT=\frac{1}{\sqrt{n}} \sum_{i=1}^n\textbf{x}_iy_i .
	\end{equation}
	As shown in \citet{he2022ghostknockoff}, when the Gaussian model $\textbf{X}\sim \text{N}(\mathbf{0},\boldsymbol{\Sigma})$ is assumed to generate knockoff copies $(\mathbf{\widetilde{X}}^{1},\dots \mathbf{\widetilde{X}}^{M})$ of features following the method of \citet{james},  corresponding knockoff copies of $Z$-scores $\mathbf{\widetilde{Z}}^{1},\dots \mathbf{\widetilde{Z}}^{M}$ that
	\begin{equation}\label{zscore_knockoff}
		\mathbf{\widetilde{Z}}^{m}=(\widetilde{Z}^{m}_1,\ldots,\widetilde{Z}^{m}_p)^\dT=\frac{1}{\sqrt{n}} \sum_{i=1}^n\widetilde{\textbf{x}}^{m}_iy_i ,\quad \text{for }m=1,\ldots,M,
	\end{equation}
	satisfy 
	\begin{equation}
		\mathbf{\widetilde{Z}}^{1:M}=
		\mathbf{P}\textbf{Z}+\widetilde{\textbf{E}},\label{Knockoff_conditional_Z}
	\end{equation}
	where $\mathbf{\widetilde{Z}}^{1:M}=(\widetilde{Z}^{1}_1,\ldots,\widetilde{Z}^{1}_p,\ldots,\widetilde{Z}^{M}_1,\ldots,\widetilde{Z}^{M}_p)^\dT$, $\widetilde{\textbf{E}}\sim \text{N}\left(\mathbf{0},\textbf{V}\right)$,
	\begin{equation}\label{equation PM}
		\mathbf{P}=\left(\begin{array}{c}\mathbf{I}-\mathbf{D} \mathbf{\Sigma}^{-\mathbf{1}} \\ \vdots \\ \mathbf{I}-\mathbf{D} \mathbf{\Sigma}^{-1}\end{array}\right) \text{, }\quad \mathbf{V}=\left(\begin{array}{cccc}\mathbf{C} & \mathbf{C}-\mathbf{D} & \cdots & \mathbf{C}-\mathbf{D} \\ \mathbf{C}-\mathbf{D} & \mathbf{C} & \ldots & \mathbf{C}-\mathbf{D} \\ \vdots & \vdots & \ddots & \vdots \\ \mathbf{C}-\mathbf{D} & \mathbf{C}-\mathbf{D} & \ldots & \mathbf{C}\end{array}\right),
	\end{equation}
 $\mathbf{D}=\operatorname{diag}\left(s_1, \ldots, s_p\right)\succeq 0$ is a diagonal matrix and $\mathbf{C}=2 \mathbf{D}-\mathbf{D} \mathbf{\Sigma}^{-\mathbf{1}} \mathbf{D}$. 
	In this article, we use the SDP construction of 
	$\mathbf{D}=\{s_1,\ldots,s_p\}$ 
	by solving the optimization problem,
	\begin{equation}\label{SDP}
		\min\sum_{j=1}^p\left|1-s_j\right|, \quad\text { s.t. }\begin{cases}
			\frac{M+1}{M} \mathbf{\Sigma-D} \succeq 0, \\
			s_j \geq 0,\quad 1 \leq j \leq p.
		\end{cases}
	\end{equation}
	In practice, the number of knockoff copies ($M$) is determined by the target FWER level as details in Section \ref{process}.

	With the convention that $\mathbf{\widetilde{Z}}^{0}=\textbf{Z}$, we show in Appendix \ref{Proof:exchangeability} that under Gaussian model $\textbf{X}\sim \text{N}(\mathbf{0},\boldsymbol{\Sigma})$, $Z$-scores $\mathbf{\widetilde{Z}}^{0:M}=(\widetilde{Z}^{0}_1,\ldots,\widetilde{Z}^{0}_p,\ldots,\widetilde{Z}^{M}_1,\ldots,\widetilde{Z}^{M}_p)^\dT$ possess the  extended exchangeability property with respect to $\{H_j:j=1,\ldots,p\}$,
	\begin{equation}\label{exchangeablility}
		\mathbf{\widetilde{Z}}^{0:M}_{swap(\boldsymbol{\sigma})}\,{\buildrel d \over =}\,\mathbf{\widetilde{Z}}^{0:M},
	\end{equation}
where
	\begin{itemize}
		\item $\mathbf{\widetilde{Z}}^{0:M}_{swap(\boldsymbol{\sigma})}=(\widetilde{Z}^{\sigma_1(0)}_1,\ldots,\widetilde{Z}^{\sigma_p(0)}_p,\ldots,\widetilde{Z}^{\sigma_1(M)}_1,\ldots,\widetilde{Z}^{\sigma_p(M)}_p)^\dT$;
		\item $\boldsymbol{\sigma}=\{\sigma_1,\ldots,\sigma_p\}$ is a family of permutations on $\{0,1,\ldots,M\}$;
		\item $\sigma_j$ is any permutation on $\{0,1,\ldots,M\}$ if $H_j$ is true and is the identity permutation if $H_j$ is false ($j=1,\ldots,p$).
	\end{itemize}

	Given $Z$-scores and their knockoff copies $\mathbf{\widetilde{Z}}^{0},\dots \mathbf{\widetilde{Z}}^{M}$, we then compare $\widetilde{Z}_j^0$ with its knockoff copies to obtain the knockoff statistics of feature  $X_j$ as
	\begin{equation}
		\label{kappa_and_tau}	
		\kappa_j=\underset{0 \leq m \leq M}{\arg \max } \text{ } (\widetilde{Z}_j^m)^2,\quad \tau_j=(\widetilde{Z}_j^{\kappa_j})^2-\underset{m\neq \kappa_j}{\operatorname{median }} (\widetilde{Z}_j^m)^2,\quad j=1,\ldots,p.
	\end{equation}
	Here, $\kappa_j$ and $\tau_j$ are generalization of the sign and the magnitude to feature statistic $W_j$ in \citet{candes2018panning} respectively. Similar to existing works of multiple knockoffs \citep{james,He2021,he2022ghostknockoff}, we also have $\kappa_j$'s corresponding to null features are independent uniform random variables as provided in  Lemma \ref{lemma1}.
	
	\begin{lemma} \label{lemma1}
		If $\{(\kappa_j,\tau_j)|j=1,\ldots,p\}$ are generated from (\ref{kappa_and_tau}) with $Z$-scores $\mathbf{\widetilde{Z}}^{0},\dots \mathbf{\widetilde{Z}}^{M}$ possessing the extended exchangeability property (\ref{exchangeablility}), $\{\kappa_j|j\in \mathcal{H}_0\}$ are independent uniform random variables   on   $\{0,1, \ldots, M\}$ conditional on $\{\kappa_j|j\notin \mathcal{H}_0\}$ and $\{\tau_j|j=1,\ldots,p\}$.
	\end{lemma}

	Proof of Lemma \ref{lemma1} is provided in Appendix \ref{proof:lemma1}, which is analogous to the proof of Lemma 3.2 in \citet{james}. However, different to the work of \citet{james}, our Lemma \ref{lemma1} directly utilizes the extended exchangeability property of $Z$-scores $\mathbf{\widetilde{Z}}^{0},\dots \mathbf{\widetilde{Z}}^{M}$.

	\subsection{Inference with FWER Control}\label{process}
	
	Given knockoff statistics $\{(\kappa_j,\tau_j)|j=1,\ldots,p\}$ generated from (\ref{kappa_and_tau}), the next step is developing a FWER filter to reject as many $H_j$'s as possible  such that (a) evidences against rejected $H_j$'s are strong and (b) the probability of rejecting at least one true $H_j$ (FWER) is bounded by the target level $\alpha$. Without loss of generality, we reindex all features such that $\tau_j$'s satisfy $\tau_1\geq \tau_2\geq \cdots\geq \tau_p$ as shown in Figure \ref{fig:fwerfilter}. Recalling that $\kappa_j$ and $\tau_j$ respectively measure whether the original feature $X_j$ is more important than its knockoff copies and the overall importance of $X_j$ and its knockoff copies, we can achieve (a) by including the first several $H_j$'s in the rejection set $\mathcal{R}$ whose $\kappa_j$'s are $0$ and $\tau_j$'s are large. In other words, we achieve (a) by letting $\mathcal{R}=\{j|\kappa_{j}=0\}\cap\{1,\ldots,T\}$ with the threshold $T$ as large as possible. However, as $T$ increases, $\mathcal{R}$ becomes larger, making it more likely to have true $H_j$'s in $\mathcal{R}$ and leading to violation of FWER control in (b). Thus, we propose the FWER filter (Algorithm \ref{FWER_filter}) to obtain the rejection set such that (a) and (b) are achieved simultaneously.
	
	{\begin{algorithm}[t]
			\caption{{FWER filter}}\label{FWER_filter}
			\begin{algorithmic}[1]
				\STATE \textbf{Input:} Knockoff statistics $\{(\kappa_j,\tau_j)|j=1,\ldots,p\}$, the number of knockoff copies $M$ and the target FWER level $\alpha$.
				\STATE Reindex all features such that $\tau_1\geq \tau_2\geq \cdots\geq \tau_p$.
				\STATE Compute the largest $v$ such that 
				\begin{equation}\label{selection}
					\mathbb{P}\{U\geq 1|U\sim \text{NB}(v,1/(M+1))\}\leq \alpha.
				\end{equation}
				\STATE Initialize the threshold $T=0$.
				\REPEAT
				\STATE $T=T+1$.
				\STATE Compute the rejection set $\mathcal{R}=\{j|\kappa_{j}=0\}\cap\{1,\ldots,T\}$.
				\UNTIL{$T=p$ or there are $v$ nonzero elements in $\{\kappa_1,\ldots,\kappa_T\}$.}
				\STATE \textbf{Output:} $\mathcal{R}$.
			\end{algorithmic}
	\end{algorithm}}

	As the threshold $T$ increases from $1$ to $p$, Algorithm \ref{FWER_filter} sequentially includes $H_T$ in the rejection set $\mathcal{R}$ if the corresponding $\kappa_T$ is zero (steps 6-7). Such a procedure terminates when the $v$-th nonzero $\kappa_T$'s  is met or $T=p$ (step 8), where $v$ is determined by step 3 to guarantee FWER control at the level $\alpha$. The rationale underlying the selection rule (\ref{selection}) comes from Lemma \ref{lemma1}, which suggests $\kappa_j$'s for all true $H_j$'s are independent uniform random variables on $\{0,1, \ldots, M\}$. That is to say, indicators $I(\kappa_j=0)$'s independently follow binomial distribution $\text{B}(1,1/(1+M))$ for all true $H_j$'s. Based on the connection of binomial distribution and negative binomial distribution, we have among all true $H_j$'s, the number of $\kappa_j$'s that equal zero before the $v$-th nonzero $\kappa_j$ follows negative binomial distribution $\text{NB}(v,1/(M+1))$. Given that the rejection set $\mathcal{R}=\{j|\kappa_{j}=0\}\cap\{1,\ldots,T\}$ includes $H_j$'s with $\kappa_j=0$ and $j\leq T$, if the inclusion procedure (steps 6-7) stops when we meet the $v$-th nonzero $\kappa_j$ among all true $H_j$'s, we have the number of false discoveries (or the number of zero $\kappa_j$'s of true $H_j$ in $\mathcal{R}$) follows $\text{NB}(v,1/(M+1))$. Thus, to control FWER under $\alpha$ with power as large as possible, we should select $v$ as large as possible such that (\ref{selection}) stands. In practice, as whether $H_j$'s are true or not is unknown, Algorithm \ref{FWER_filter} terminates the inclusion procedure (steps 6-7) when we meet the $v$-th nonzero $\kappa_j$ (as shown in Figure \ref{fig:fwerfilter}) or $T=p$, which does not violate the FWER control.
	
	\begin{figure}[t]
		\centering
		\begin{subfigure}{0.48\linewidth}
			\includegraphics[width=\linewidth]{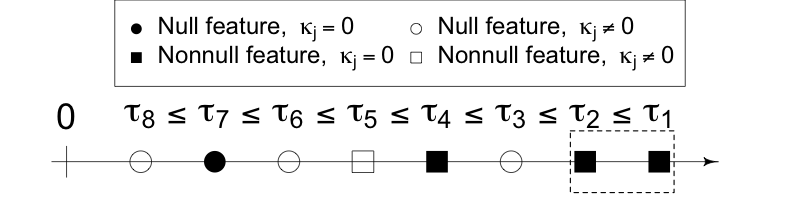}
			\caption{$v=1$.}
		\end{subfigure}
		\begin{subfigure}{0.48\linewidth}
			\includegraphics[width=\linewidth]{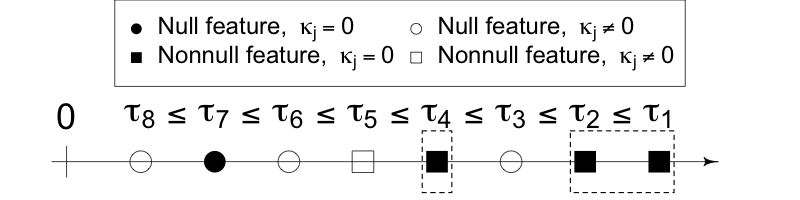}
			\caption{$v=2$.}
		\end{subfigure}
		\caption{Graphical illustration of the rejection set $\mathcal{R}$ computed by Algorithm \ref{FWER_filter} which terminates the inclusion procedure (steps 6-7) when we meet the $v$-th nonzero $\kappa_j$ for different $v$.}
		\label{fig:fwerfilter}
	\end{figure}
	
	As the selection rule (\ref{selection}) is valid in FWER control for any target level $\alpha$ and the number of knockoff copies $M$, we can utilize it to determine $M$ for knockoff copies generation in Section \ref{GK}. If $M$ is too small,  selection rule (\ref{selection}) would obtain $v=0$, resulting in empty rejection set and no power. Taking such issue into account, we suggest choosing the smallest $M$ such that selection rule (\ref{selection}) obtains $v=1$. For example, when the target FWER level is $\alpha=0.05$, we have selection rule (\ref{selection}) obtains $v=0$ for $M=1,\ldots,18$ and $v=1$ for $M=19$. As a result, we use $M=19$ (which leads to $v=1$) throughout all numerical experiments with the target FWER level $\alpha=0.05$ in this article.
	
	\begin{rmk}
		Similar to methods in \citet{janson} and \citet{zhimei}, Algorithm \ref{FWER_filter} also rejects $H_j$'s with FWER control. However, there exist substantial differences between our method and existing ones. 
		
		In  \citet{janson} where $M$ is fixed as $1$, they propose two selection rules of $v$ to control $k$-FWER under the target level $\alpha$. The first one is the deterministic rule which computes the largest $v$ such that 
		$$
			\mathbb{P}\{U\geq k|U\sim \text{NB}(v,1/2)\}\leq \alpha,
		$$
		analogously to (\ref{selection}). However, when we target to control FWER (or equivalently $k=1$)  at the target level $\alpha<1/2$, the deterministic rule obtains $v=0$, resulting in empty rejection set and no power. To address this issue, they further propose the stochastic rule to select $v$ where
		$$
		\begin{aligned}
			v=&\begin{cases}
			v_\alpha,&\text{\rm w.p. } 1-q,\\
			v_\alpha+1,&\text{\rm w.p. } q,
		\end{cases}\\\text{\rm such that }&\begin{cases}
		\mathbb{P}\{U\geq k|U\sim \text{NB}(v_\alpha,1/2)\}\leq \alpha\leq \mathbb{P}\{U\geq k|U\sim \text{NB}(v_\alpha+1,1/2)\},\\
		(1-q)\cdot\mathbb{P}\{U\geq k|U\sim \text{NB}(v_\alpha,1/2)\}+q\cdot\mathbb{P}\{U\geq k|U\sim \text{NB}(v_\alpha+1,1/2)\}=\alpha.
	\end{cases}
		\end{aligned}
		$$
		However, when we target to control FWER (or equivalently $k=1$) at the target level $\alpha<1/2$, the stochastic rule obtains $v_\alpha=0$ and $q=2\alpha$. In other words, the procedure of \citet{janson} only makes rejections ($v=1$) with probability $2\alpha$ and thus its power is bounded by $2\alpha$. As a result, the procedure of \citet{janson} would suffer great power loss in practice, where people commonly target to control FWER under $0.01$, $0.05$ or $0.1$. In contrast, such a power bound does not exist for our method where $M$ is selected as the smallest one that leads to $v=1$ via (\ref{selection}).
		
		Although derandomized knockoffs \citep{zhimei} are able to control FWER, such a control is achieved indirectly by controlling per family error rate {\rm($\text{PFER}=\mathbb{E}(V)$)}. According to the Markov inequality, $$\text{\rm PFER}=\mathbb{E}(V)\geq \mathbb{P}(V\geq 1)=  \text{\rm FWER},$$ PFER control is more stringent than FWER control. Thus, as shown in experimental results in Section \ref{section4}, there unavoidably exists a gap between the empirical FWER and the target level, leading to power loss. In contrast, as our FWER filter directly controls FWER, its empirical FWER concentrates at the target level and thus has higher power.
	\end{rmk}

	\section{Computational Strategy}\label{section3}

	\subsection{Efficient Generation of Knockoff Copies}
	
	To generate knockoff copies of $Z$-scores $\mathbf{\widetilde{Z}}^{1},\dots \mathbf{\widetilde{Z}}^{M}$ via (\ref{Knockoff_conditional_Z}), it is crucial to sample a $pM$-dimensional random vector $\widetilde{\textbf{E}}\sim \text{N}(\mathbf{0,V})$. 
	However, directly doing so usually requires performing Cholesky decomposition of the $pM\times pM$ matrix $\mathbf{V}$, which is computationally intensive especially when $M$ is large. For example, as the target FWER level is fixed as $\alpha=0.05$ throughout this article, we need to generate $M=19$ knockoff copies of $Z$-scores and thus need to implement Cholesky decomposition of a $19000\times 19000$ matrix when we are inferring $p=1000$ features. To circumvent such an obstacle, we propose an efficient algorithm to generate $\widetilde{\mathbf{E}}$ as detailed in the following.
	
	Noting that the covariance matrix $\textbf{V}$ can be decomposed as $\textbf{V}=\textbf{V}_1+\textbf{V}_2$ where \begin{equation}\label{V1V2}
		\textbf{V}_1=\begin{pmatrix}
			\textbf{C}-\frac{M-1}{M}\textbf{D}&\textbf{C}-\frac{M-1}{M}\textbf{D}&\cdots&\textbf{C}-\frac{M-1}{M}\textbf{D}\\
			\textbf{C}-\frac{M-1}{M}\textbf{D}&\textbf{C}-\frac{M-1}{M}\textbf{D}&\cdots&\textbf{C}-\frac{M-1}{M}\textbf{D}\\
			\vdots&\vdots&\ddots&\vdots\\
			\textbf{C}-\frac{M-1}{M}\textbf{D}&\textbf{C}-\frac{M-1}{M}\textbf{D}&\cdots&\textbf{C}-\frac{M-1}{M}\textbf{D}\\
		\end{pmatrix}\text{, }\quad \textbf{V}_2=\begin{pmatrix}
			\frac{M-1}{M}\textbf{D}&-\frac{1}{M}\textbf{D}&\cdots&-\frac{1}{M}\textbf{D}\\
			-\frac{1}{M}\textbf{D}&\frac{M-1}{M}\textbf{D}&\cdots&-\frac{1}{M}\textbf{D}\\
			\vdots&\vdots&\ddots&\vdots\\
			-\frac{1}{M}\textbf{D}&-\frac{1}{M}\textbf{D}&\cdots&\frac{M-1}{M}\textbf{D}\\
		\end{pmatrix},
	\end{equation} 
and $$\begin{aligned}
	\textbf{C}-\frac{M-1}{M}\textbf{D}=&\frac{M+1}{M}\textbf{D}-\textbf{D}\boldsymbol{\Sigma}^{-1}\textbf{D}
	=\textbf{D}^{1/2}\Biggl\{\frac{M+1}{M}\textbf{I}-\textbf{D}^{1/2}\boldsymbol{\Sigma}^{-1}\textbf{D}^{1/2}\Biggr\}\textbf{D}^{1/2}
	\succeq 0,
\end{aligned}$$
we can sample $\widetilde{\textbf{E}}\sim \text{N}(\mathbf{0,V})$ as the sum
$$\begin{aligned}
	&\widetilde{\mathbf{E}}=
	\widetilde{\mathbf{E}}_1+\widetilde{\mathbf{E}}_2,\quad
	\text{where }\widetilde{\mathbf{E}}_1
	\sim \text{N}(\textbf{0},\textbf{V}_1)\text{ and } \widetilde{\mathbf{E}}_2
	\sim \text{N}(\textbf{0},\textbf{V}_2).
\end{aligned}$$ 
Based on the block structure of $\textbf{V}_1$ and $\textbf{V}_2$ in (\ref{V1V2}), random vectors  $\widetilde{\textbf{E}}_1$ and $\widetilde{\textbf{E}}_2$ can be efficiently generated as follows.
	
	\begin{itemize}
		\item[$\star$] \textbf{(Generating $\widetilde{\mathbf{E}}_1$)} Noting that $\mathbf{V}_1$ is a $M\times M$ block matrix with the common positive semi-definite block $\textbf{C}-\frac{M-1}{M}\textbf{D}$,
		we generate $\widetilde{\mathbf{E}}_1=((\textbf{L}\widetilde{\textbf{e}}_1)^\dT,\ldots,(\textbf{L}\widetilde{\textbf{e}}_1)^\dT)^\dT$, where $\widetilde{\textbf{e}}_1\sim \text{N}(\mathbf{0},\textbf{I})$ and $\textbf{L}$ is the solution of the $p\times p$ Cholesky decomposition equation,
		$$\textbf{C}-\frac{M-1}{M}\textbf{D}=\textbf{L}\textbf{L}^\dT.$$
		
		\item[$\star$] \textbf{(Generating $\widetilde{\mathbf{E}}_2$)} With i.i.d. samples  $\widetilde{\textbf{e}}^{1}_2,\widetilde{\textbf{e}}^{2}_2,\ldots,\widetilde{\textbf{e}}^{M}_2\sim\text{N}(\mathbf{0},\textbf{D})$, we  
		compute $\widetilde{\mathbf{E}}_2=(	(\widetilde{\textbf{e}}^{1}_2-\overline{\textbf{e}}_2)^\dT,(\widetilde{\textbf{e}}^{2}_2-\overline{\textbf{e}}_2)^\dT,\ldots,(\widetilde{\textbf{e}}^{M}_2-\overline{\textbf{e}}_2)^\dT)^\dT$, where $\overline{\textbf{e}}_2=M^{-1}\sum_{m=1}^{M}\widetilde{\textbf{e}}^{m}_2$.

		%
	\end{itemize}

	Details of generating knockoff copies of $Z$-scores $\mathbf{\widetilde{Z}}^{1},\dots \mathbf{\widetilde{Z}}^{M}$ are summarized in Algorithm \ref{alg:eff}. Compared with the trivial approach that relies on Cholesky decomposition of $\textbf{V}$, Algorithm \ref{alg:eff} decreases the computational complexity from $O(M^3p^3)$ to $O(p^3)$.
	{\begin{algorithm}
			\caption{Efficient generation of knockoff copies of $Z$-scores.}\label{alg:eff}
			\begin{algorithmic}[1]
				\STATE \textbf{Input:} The covariance matrix $\mathbf{\Sigma}$ of $\textbf{X}$ and the number of knockoff copies ${M}$.
				\STATE Compute the diagonal matrix $\textbf{D}$ by solving the optimization problem (\ref{SDP}).
				\STATE Compute $\textbf{I}-\textbf{D}\boldsymbol{\Sigma}^{-1}$ and $\textbf{P}$ via (\ref{equation PM}).
				\STATE Perform Cholesky decomposition 
				$$\textbf{C}-\frac{M-1}{M}\textbf{D}=\textbf{L}\textbf{L}^\dT,$$ where $\mathbf{C}=2\mathbf{D}-\mathbf{D\Sigma^{-1}D}$.
				\STATE Sample a $p$-dimensional random vector $\widetilde{\textbf{e}}_1\sim\text{MVN}(\mathbf{0},\textbf{I})$.
				\STATE Compute $\widetilde{\mathbf{E}}_1=((\textbf{L}\widetilde{\textbf{e}}_1)^\dT,\ldots,(\textbf{L}\widetilde{\textbf{e}}_1)^\dT)^\dT$.
				\STATE Draw i.i.d. $p$-dimensional random vectors   $\widetilde{\textbf{e}}^{1}_2,\widetilde{\textbf{e}}^{2}_2,\ldots,\widetilde{\textbf{e}}^{M}_2\sim\text{N}(\mathbf{0},\textbf{D})$.
				\STATE Compute $\overline{\textbf{e}}_2=M^{-1}\sum_{m=1}^{M}\widetilde{\textbf{e}}^{m}_2$;
				\STATE Compute $\widetilde{\mathbf{E}}_2=(	(\widetilde{\textbf{e}}^{1}_2-\overline{\textbf{e}}_2)^\dT,(\widetilde{\textbf{e}}^{2}_2-\overline{\textbf{e}}_2)^\dT,\ldots,(\widetilde{\textbf{e}}^{M}_2-\overline{\textbf{e}}_2)^\dT)^\dT$.
				\STATE Compute $\widetilde{\mathbf{E}}=
				\widetilde{\mathbf{E}}_1+\widetilde{\mathbf{E}}_2$.
				\STATE \textbf{Output:}  
				$\mathbf{\widetilde{Z}}^{1:M}=
				\mathbf{P}\textbf{Z}+\widetilde{\textbf{E}}$.
			\end{algorithmic}
	\end{algorithm}}
	
		\subsection{Evaluation of Computational Efficiency}\label{section:time}

	To empirically evaluate the computational efficiency of our method in generating knockoff copies of $Z$-scores and performing inference, we conduct experiment to compare our method with derandomized knockoffs (where $M=1$, \citealp{zhimei}) and the trivial approach via (\ref{Knockoff_conditional_Z}) with Cholesky decomposition of $\textbf{V}$.  For different number of features $p =50, 100, 200$ and $500$, $200$ datasets are generated and inferred. Within each dataset, $n=1000$ i.i.d. samples $\{(\textbf{x}_i,y_i):i=1,\ldots,n\}$ are generated from the linear model
	\begin{equation} \label{eq:gen}
		Y=  \boldsymbol{\beta}^\dT\mathbf{X} +\epsilon, \quad \epsilon \sim N(0,1),
	\end{equation}
	where $p$-dimensional feature vector $\mathbf{X} \sim \text{N}\big(0,\mathbf{\Sigma}=(0.25^{|i-j|})_{p\times p}\big)$. With a fixed number of nonnull features $|\mathcal{H}_1|=10$, locations of nonnull features are randomly drawn from $\{1,\ldots,p\}$ with $\beta_j$'s uniformly distributed in $\{5/\sqrt{n},-5/\sqrt{n}\}$ if $j\in \mathcal{H}_1$ and $\beta_j=0$ for null features. 
	With the target FWER level $0.05$, the hyperparameter of derandomized knockoffs are selected as the number of repetitions $M_{deran}=50$ using codes of \citet{zhimei}.
	
	
	Specifically, we compare computation time of three approaches in 
	\begin{itemize}
		\item[(a)] Cholesky decomposition (derandomized knockoffs: $2 \mathbf{D}-\mathbf{D} \mathbf{\Sigma}^{-\mathbf{1}} \mathbf{D}$; trivial approach: $\textbf{V}$; our method: $\textbf{C}-\frac{M-1}{M}\textbf{D}$);
		\item[(b)] Sampling knockoff copies of $Z$-scores, including sampling $\widetilde{\textbf{E}}$ and computing $\mathbf{\widetilde{Z}}^{1:M}=
		\mathbf{P}\textbf{Z}+\widetilde{\textbf{E}}$ ($M_{\text{deran}}$ repetitions for derandomized knockoffs with $M=1$);
		\item[(c)] Inference (Algorithm \ref{FWER_filter} for both trivial approach and our method; $M_{\text{deran}}$ repetitions for derandomized knockoffs).
	\end{itemize}

	Here, we omit the comparison in performing precalculation of $\textbf{D}$ and $\textbf{I}-\textbf{D}\boldsymbol{\Sigma}^{-1}$ because the total complexity of this step is the same for different methods.
	Average computational time of different approaches under different dimensions $p$ is summarized in Table \ref{Table:time} (CPU: Apple M2 Pro, 3.5 GHz; RAM: 16GB). Generally, our method has the least computational time. Compared to derandomized knockoffs, our method has comparable computational complexity in Cholesky decomposition and outperforms in sampling knockoff copies and inference. The reason is that $M_{\text{deran}}$ samplings of $p$-dimensional random vector and $M_{\text{deran}}$ inferences are needed in derandomized knockoffs, leading to $\Omega(M_{\text{deran}} p^2)$ and $\Omega(M_{\text{deran}} p\log p)$ computational costs respectively. Compared to the trivial approach, our method do relieve the computational burden of Cholesky decomposition from $O(M^3p^3)$ to $O(p^3)$. As a result, the overall computational complexity of Algorithm \ref{alg:eff} is the smallest.
	

	\begin{table}[t]
		\centering
		\caption{Average computational times (s) of derandomized knockoffs (with $M=1$ and $M_{deran}=50$), the trivial approach and our method (with $M=19$) in Cholesky decomposition, sampling knockoff copies  and inference
			with sample size $n=1000$ (CPU: Apple M2 Pro, 3.5 GHz; RAM: 16GB). }\label{Table:time}
		\begin{tabular}{llrrrr}
			\toprule
			\multicolumn{2}{c}{$p$}&$50$&$100$&$200$&$500$\\
			\midrule
			\multirow{5}{*}{{Derandomized knockoffs}}						
			& Cholesky decomposition & \textbf{0.00145} & 0.00170 & 0.00291 & \textbf{0.01603}\\ 
			&Sampling knockoff copies& 0.03308 & 0.03339 & 0.03652 & 0.05816 \\ 
			&Inference & 0.07840 & 0.13895 & 0.26577 & 0.61804 \\ 
			\cline{2-6}
			&Total time& 0.11293& 0.17404 &0.30520 &0.69223  \\ 
			\midrule
			\multirow{4}{*}{{Trivial approach}}
 		    	&Cholesky decomposition&  0.16140 &  0.66588 &  4.46985 & 82.07046 \\ 
  		    	&Sampling knockoff copies&  0.00083 &  0.00277 &  0.00916 &  0.06004 \\ 
  		        &Inference &  0.00168 &  0.00340 &  0.00517 &  \textbf{0.01252 }\\ 
			\cline{2-6}
			&{Total time}&0.16391 & 0.67205 &  4.48418 & 82.14302 \\ 
%
%
			\midrule
			
			\multirow{4}{*}{{Our method}}
			  & Cholesky decomposition & \textbf{0.00145} & \textbf{0.00168} & \textbf{0.00288} & 0.01696 \\ 
			  &Sampling knockoff copies& \textbf{0.00073} & \textbf{0.00074} & \textbf{0.00091} & \textbf{0.00165} \\ 
			  &Inference & \textbf{0.00162} & \textbf{0.00279} & \textbf{0.00516} & 0.01263 \\ 
				\cline{2-6}
			&{Total time}&\textbf{0.00380}& \textbf{0.00521}& \textbf{0.00895} &\textbf{0.03124} \\ 
			\bottomrule
		\end{tabular}
	\end{table}


	%
	%
	
	\section{Experiments} \label{section4}
	
	We conduct extensive experiments on synthetic data (i) to evaluate the performance of the proposed FWER filter with GhostKnockoff in both FWER control and power and (ii) to compare with existing knockoff-based methods that also control FWER.
	 With the target FWER level $\alpha=0.05$,  we set the number of knockoff copies $M=19$ and the hyperparameter $v=1$ for our method.  For comparison, we implement the procedure of \citet{janson} and derandomized knockoffs \citep{zhimei}, which repeatedly implement the procedure of \citet{janson} (with $M=1$) for $M_{deran}$ times and return the final selection set as
	\begin{equation}
		\label{selection_freq}
	{\mathcal{R}}=\{j|\Pi_j\geq \eta\},\quad \text{where }\Pi_j= \frac{1}{M_{deran}}\sum_{m=1}^{M_{deran}} \mathbf{I}\{j \in {\mathcal{R}}^m\},
	\end{equation} with rejection sets $\{{\mathcal{R}}^m|m=1,\ldots ,M_{deran}\}$ obtained in each repetition. We choose $M_{deran}=50$ and $\eta=0.99$ using codes of \citet{zhimei}.
	%
	%
	For all methods, we use the SDP construction (\ref{SDP})  of $\textbf{D}$ and $Z$-scores in (\ref{zscore}) as feature importance scores.

	\subsection{Correlation Structure}\label{section:structure}

	To investigate how the proposed FWER filter with GhostKnockoff performs under different correlation structures with comparison to existing knockoff-based methods, we conduct experiments of $500$ randomly generated datasets under both the compound symmetry correlation structure and the $\text{AR}(1)$ correlation structure of features. For each dataset, $n=500$ observations are generated by drawing i.i.d. samples of $100$-dimensional features $\textbf{x}_1,\ldots,\textbf{x}_n\sim \text{N}(0,\mathbf{\Sigma})$, where  
	\begin{equation}\label{eq:correlation}
		\boldsymbol{\Sigma}=(\sigma_{ij})_{p \times p} \quad \text { where } \quad \sigma_{ij}= \begin{cases}\rho^{I(i \neq j)}, & \text { (compound symmetry structure); }\\\rho^{|i-j|}, & \text { ({AR}(1) structure) ,} \end{cases}
	\end{equation}
	and simulating responses $y_1,\ldots,y_n$ from the linear model (\ref{eq:gen}). With a fixed number of nonzero elements $|\mathcal{H}_1| =5$ in the parameter vector $\boldsymbol{\beta}=(\beta_1,\ldots,\beta_p)^\dT$, locations of these nonzero elements are uniformly sampled from $\{1,\ldots,p\}$ without replacement while their values are i.i.d. uniform variables in $\{A/\sqrt{n},-A/\sqrt{n}\}$. To exhibit the complete power curve, we let the amplitude $A$ of signals range in  $\{1,2,\ldots,20\}$ under compound symmetry structure and $\{1,2,\ldots,25\}$ under $\text{AR}(1)$ structure respectively.
	
		\begin{figure}
		\begin{subfigure}[b]{0.49\textwidth}
			\centering
			\includegraphics[width=\textwidth]{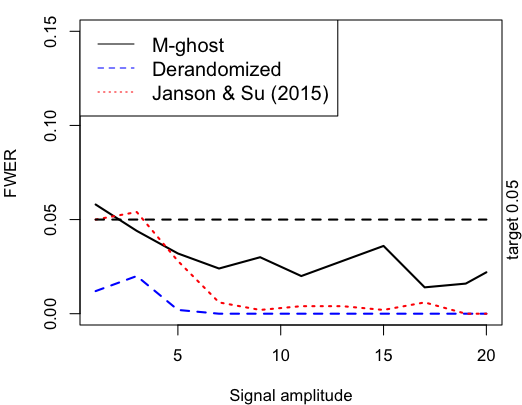}
		\end{subfigure}
		\begin{subfigure}[b]{0.49\textwidth}
			\centering
			\includegraphics[width=\textwidth]{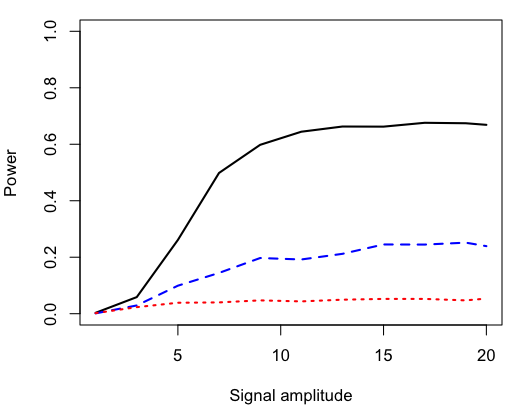}
		\end{subfigure}
		\begin{subfigure}[b]{0.49\textwidth}
		\centering
		\includegraphics[width=\textwidth]{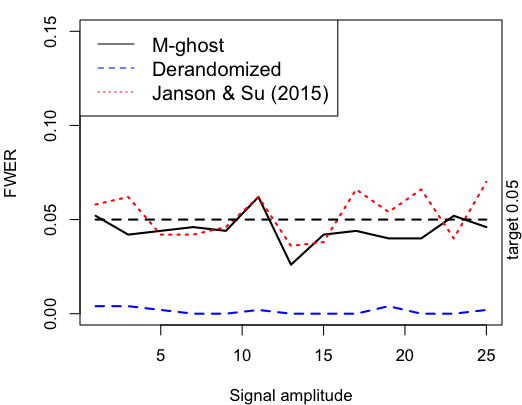}
	\end{subfigure}
\
	\begin{subfigure}[b]{0.49\textwidth}
		\centering
		\includegraphics[width=\textwidth]{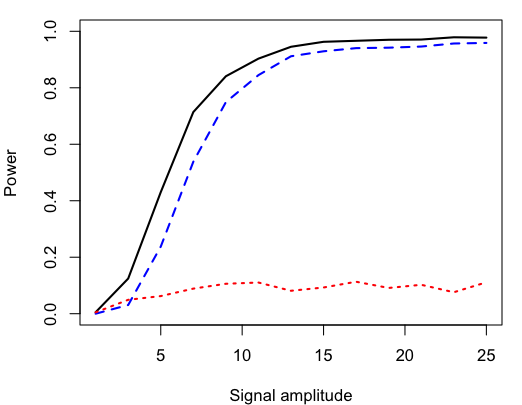}
	\end{subfigure}
		\caption{Empirical FWER and power over $500$ simulated datasets of the proposed FWER filter with GhostKnockoff, derandomized knockoffs and the procedure of \citet{janson} with respect to different signal amplitudes $A$ under compound symmetry (top panel) and AR(1) (bottom panel) correlation structure, sample size $n=500$, dimension $p=100$, correlation strength $\rho=0.5$ and the number of nonnull features $|\mathcal{H}_1|=5$.}
		\label{fig:generic_cs}
	\end{figure}

	With the correlation coefficient $\rho=0.5$ and the target FWER level $0.05$, average power of different methods over $500$ datasets under compound symmetry structure and $\text{AR}(1)$ structure is presented in Figure \ref{fig:generic_cs}. With FWER controlled under the target level, the power of the proposed FWER filter with GhostKnockoff and derandomized knockoffs increase as the signal amplitude $A$ grows under both structures, while the procedure of \citet{janson} maintains a low power. The reason is that the procedure of \citet{janson} only makes rejection with probability $2\times0.05=0.1$ as shown in Section \ref{process}, making the power bounded by $0.1$. Although the proposed FWER filter with GhostKnockoff dominates derandomized knockoffs in power under both structures, its power can only reach $1$ as $A$ grows under $\text{AR}(1)$ structure and converges to a limit smaller than 1 under the compound symmetry structure. More results under different dimensions and correlation strengths are also provided in Appendix \ref{Supp41} with similar conclusions.


	\subsection{Correlation Strength}\label{section:correlation}
	
	To further investigate how correlation strength affects the performance of the proposed FWER filter with GhostKnockoff, we generate $500$ datasets for each possible correlation strength $\rho \in \{0,0.1,\ldots,0.9\}$ under both compound symmetry and $\text{AR}(1)$ structure. With a fixed sample size $n=500$, $100$-dimensional features are sampled as i.i.d. observations $\textbf{x}_1,\ldots,\textbf{x}_n\sim \text{N}(0,\mathbf{\Sigma})$ and responses $y_1,\ldots,y_n$ are computed from the linear model (\ref{eq:gen}). With a fixed number of nonzero elements $|\mathcal{H}_1| =5$ in the parameter vector $\boldsymbol{\beta}=(\beta_1,\ldots,\beta_p)^\dT$, locations of these nonzero elements are obtained by uniformly sampling of $\{1,\ldots,p\}$ without replacement while their values are i.i.d. uniform variables in $\{A/\sqrt{n},-A/\sqrt{n}\}$. The amplitude $A$ is $6$ and $8$ under compound symmetry and  $\text{AR}(1)$ structure respectively to completely show how power varies with respect to correlation strength. To avoid the power loss phenomenon when the SDP construction of  $\mathbf{D}$ is used under the compound symmetry structure with $\rho\geq 0.5$ as discovered in \citet{Arther}, we adopt their strategy by multiplying $\mathbf{D}$ with a perturbation factor $\gamma=0.9$.
	
		\begin{figure}
		\begin{subfigure}[b]{0.49\textwidth}
			\centering
			\includegraphics[width=\textwidth]{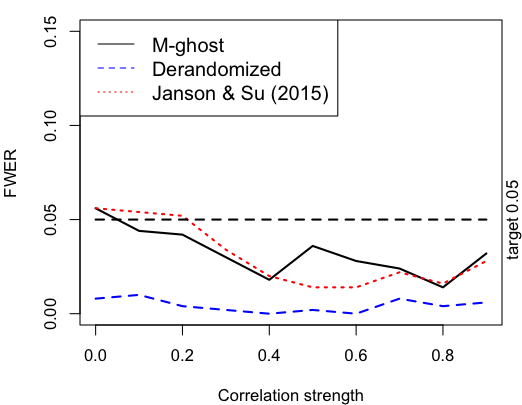}
		\end{subfigure}
		\begin{subfigure}[b]{0.49\textwidth}
			\centering
			\includegraphics[width=\textwidth]{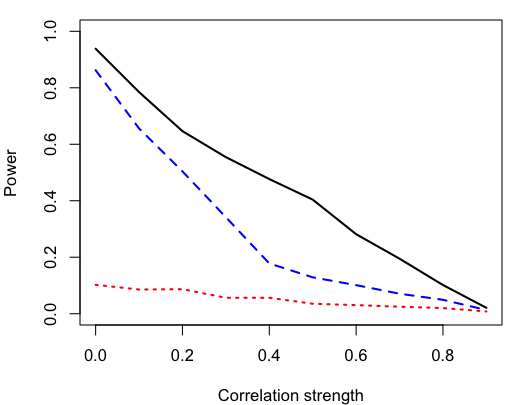}
		\end{subfigure}
			\begin{subfigure}[b]{0.49\textwidth}
		\centering
		\includegraphics[width=\textwidth]{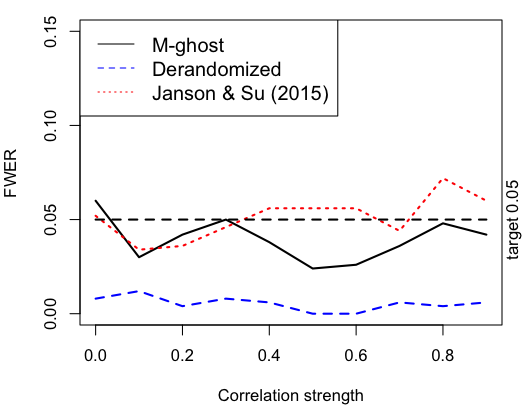}
	\end{subfigure}
\
	\begin{subfigure}[b]{0.49\textwidth}
		\centering
		\includegraphics[width=\textwidth]{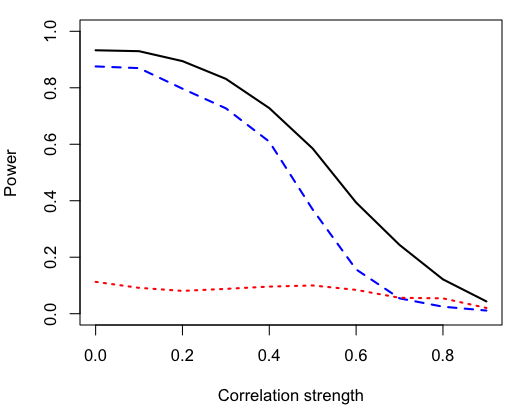}
	\end{subfigure}
		\caption{
			Empirical FWER and power over $500$ simulated datasets of the proposed FWER filter with GhostKnockoff, derandomized knockoffs and the procedure of \citet{janson} with respect to different correlation strengths $\rho$ under compound symmetry (top panel with $A=6$) and  AR(1) (bottom panel with $A=8$) correlation structure, sample size $n=500$, dimension $p=100$ and the number of nonnull features $|\mathcal{H}_1|=5$.}
		\label{fig:cs_cr}
	\end{figure}
	Empirical FWER and power of the proposed FWER filter with GhostKnockoff and other existing approaches are visualized in Figure \ref{fig:cs_cr}. Consistent with our theoretical derivation, our method can control the FWER under the target level $\alpha=0.05$. While the procedure of \citet{janson} remains powerless, both the proposed FWER filter with GhostKnockoff and  knockoffs have a decreasing power as the correlation strength $\rho$ grows. The main reason is that when $\rho$ gets larger, nearby features become more similar, making it more difficult to distinguish true signals. Even though, our method still outperforms  derandomized knockoffs, whose power has a sudden drop when $\rho\geq0.5$ due to the perturbation factor $\gamma$.

	
	\subsection{Number of Nonnull Features and Sample Size}
	
	\begin{figure}
		\begin{subfigure}[b]{0.49\textwidth}
	\centering
	\includegraphics[width=\textwidth]{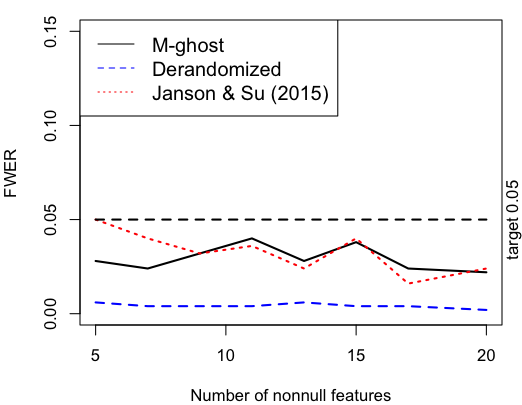}
\end{subfigure}
\begin{subfigure}[b]{0.49\textwidth}
	\centering
	\includegraphics[width=\textwidth]{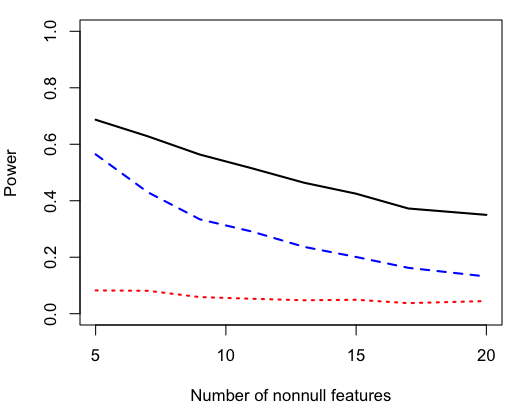}
\end{subfigure}
\begin{subfigure}[b]{0.49\textwidth}
	\centering
	\includegraphics[width=\textwidth]{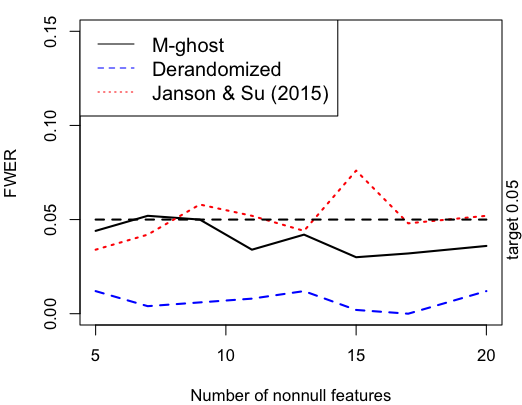}
\end{subfigure}
\
\begin{subfigure}[b]{0.49\textwidth}
	\centering
	\includegraphics[width=\textwidth]{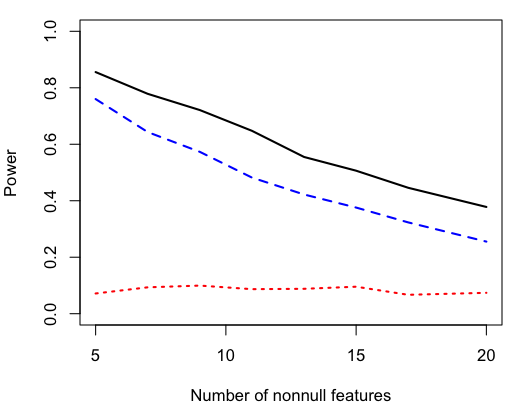}
\end{subfigure}
		\caption{Empirical FWER and power over $500$ simulated datasets of the proposed FWER filter with GhostKnockoff, derandomized knockoffs and the procedure of \citet{janson} with respect to different numbers of nonnull features $|\mathcal{H}_1|$ under compound symmetry (top panel with $A=6$) and  AR(1) (bottom panel with $A=8$) correlation structure, sample size $n=500$, dimension $p=200$ and the  correlation strength $\rho=0.25$.}
		\label{fig:ar_k}
	\end{figure}
	
	To demonstrate the empirical performance of the proposed FWER filter with GhostKnockoff under different number of nonnull features, we generate 500 datasets for each possible number of nonnull features $|\mathcal{H}_1| =5,6,\ldots,20$ under both compound symmetry (amplitude $A=6$) and AR(1) (amplitude $A=8$) structures with sample size $n=500$, dimension $p=200$ and correlation strength $\rho=0.25$. Empirical FWER and power of the proposed FWER filter with GhostKnockoff under different $|\mathcal{H}_1|$ are visualized in Figures \ref{fig:ar_k}, with comparison to derandomized knockoffs and the procedure of \citet{janson}. Similarly, we also simulate 500 datasets for sample sizes $n=100,200,500$ and $1000$ under both compound symmetry (amplitude $A=6$) and AR(1) (amplitude $A=8$) structures with $p=200$, $|\mathcal{H}_1| =5$ and $\rho=0.25$. Empirical FWER and power of the proposed FWER filter with GhostKnockoff and existing methods under different sample sizes are shown in Figures \ref{fig:cs_n}.
	
	
		\begin{figure}
		\begin{subfigure}[b]{0.49\textwidth}
	\centering
	\includegraphics[width=\textwidth]{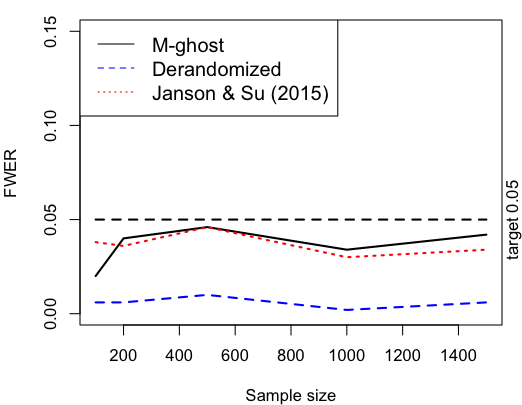}
\end{subfigure}
\begin{subfigure}[b]{0.49\textwidth}
	\centering
	\includegraphics[width=\textwidth]{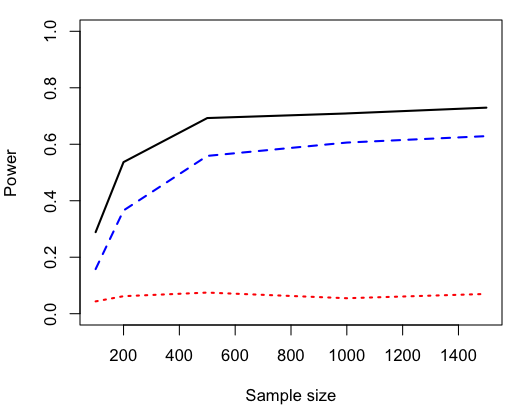}
\end{subfigure}
\begin{subfigure}[b]{0.49\textwidth}
	\centering
	\includegraphics[width=\textwidth]{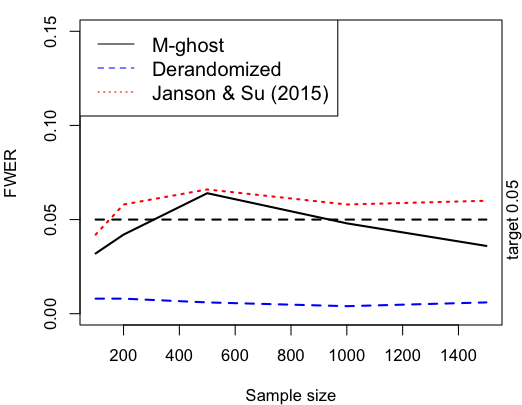}
\end{subfigure}
\
\begin{subfigure}[b]{0.49\textwidth}
	\centering
	\includegraphics[width=\textwidth]{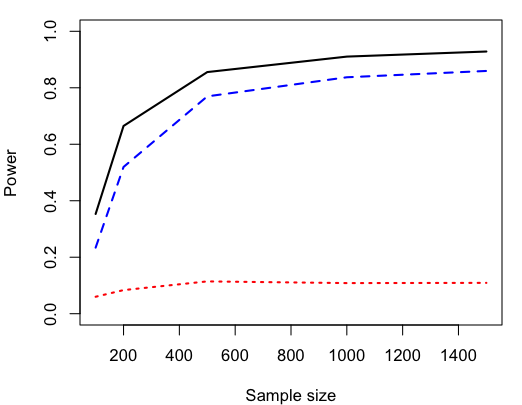}
\end{subfigure}
		\caption{Empirical FWER and power over $500$ simulated datasets of the proposed FWER filter with GhostKnockoff, derandomized knockoffs and the procedure of \citet{janson} with respect to different sample sizes $n$ under compound symmetry (top panel with $A=6$) and  AR(1) (bottom panel with $A=8$) correlation structure, dimension $p=200$, the correlation strength $\rho=0.25$ and numbers of nonnull features $|\mathcal{H}_1|=5$.}
		\label{fig:cs_n}
	\end{figure}
	
	From Figure \ref{fig:ar_k}, it is found that as the number of nonnull features $|\mathcal{H}_1| $ grows, the power of both the proposed FWER filter with GhostKnockoff and derandomized knockoffs decreases, suggesting that their ability to detect nonnull features deteriorates under dense but weak signals. As the sample size $n$ increases, the power of both methods grows to $1$ as shown in Figure \ref{fig:cs_n}, implying their consistency. In contrast, the procedure of \citet{janson} keeps a bounded power under $10\%$ throughout all scenarios. More results under different correlation strengths are provided in Appendix \ref{Supp43} with similar conclusions.

	The main reason why the power of the proposed method decreases with respect to $|\mathcal{H}_1| $ and increases with respect to $n$ is that the signal-to-noise ratio is of the order $\Omega(\sqrt{n/|\mathcal{H}_1|})$. This can be seen in the following toy example. Consider i.i.d. samples $\textbf{x}_1,\ldots,\textbf{x}_n\sim \text{N}(\textbf{0},\textbf{I})$ and responses $y_1,\ldots,y_n$ generated from the linear model (\ref{eq:gen}), where $\beta_1=\ldots=\beta_k=A/\sqrt{n}$ and $\beta_{k+1}=\cdots=\beta_p=0$ for some integer $k$. In other words, we have $|\mathcal{H}_1| =k$. By definition (\ref{zscore}) of $Z$-scores after standardizing both features and the response, we have $Z_j=\sqrt{n}\hat{\rho}_j$ where $\hat{\rho}_j$ is the empirical correlation of $X_j$ and $Y$. By (\ref{eq:gen}), we have $\text{Var}(X_j)=1$ ($j=1,\ldots,p$), $\text{Cov}(X_j,Y)=\beta_j$ and $\text{Var}(Y)=\sum_{j=1}^p\beta_j^2+1$, leading to true correlations ${\rho}_j$ as
	$$\rho_j=	\frac{\text{Cov}(X_j,Y)}{\sqrt{\text{Var}(X_j)\text{Var}(Y)}}=\begin{cases}
		\frac{A/\sqrt{n}}{\sqrt{kA^2/n+1}},& j=1,\ldots,k;\\
		0,& \text{otherwise}.
	\end{cases}$$
	By the central limit theorem of correlation coefficient \citep{Lehmann1999} that 
	$$\sqrt{n}(\hat{\rho}_j-\rho_j)\mathop{\longrightarrow}^{\text{D}}\text{N}(0,(1-\rho_{j}^2)^2),\quad \text{where }\mathop{\longrightarrow}^{\text{D}}\text{ stands for convergence in distribution,}$$
	we have that as the amplitude $A$ grows, the expected mean of $Z_j$ for nonnull features is 
	$\sqrt{n/k}\{1+o(1)\}$. Since the mean and asymptotic variance of $Z_j$ corresponding to null features remain $0$ and $1$, the signal-to-noise ratio is $\Omega(\sqrt{n/k})$. Thus, smaller $|\mathcal{H}_1|$ and larger sample size would both lead to larger signal-to-noise ratio and higher power.

	\section{Real Data Analysis}\label{section5}
	
	To investigate the empirical performance of the proposed FWER filter with GhostKnockoff in detecting genetic variants associated with the Alzheimer’s disease (AD), we apply it on the aggregated summary statistics of genetic variants from nine overlapping large-scale array-based genome-wide association studies, and whole-exome/-genome sequencing studies on individuals with European ancestry as summarized in Table \ref{Data_info}. Since AD is believed to occur when an abnormal amount of amyloid beta (a type of plasma protein) accumulates in brains \citep{Tackenberg2020}, we extract $Z$-scores of $7,963$ variants in protein quantitative trait loci (pQTL) that are reported  in \citet{Ferkingstad2021} to have significant association with the level of at least one plasma protein. For comparison, we also apply derandomized knockoffs \citep{zhimei} on the same datasets.

		\begin{table}[t]
		\centering
		\caption{Information of selected large-scale array-based genome-wide association studies, and whole-exome/-genome sequencing studies.}\label{Data_info}
		\resizebox{\columnwidth}{!}{
			\begin{tabular}{lrr}
				\toprule
				\multirow{2}{*}{Studies}&\multicolumn{2}{c}{Sample Size}\\
				\cline{2-3}
				&AD case samples&Control Sample\\
				\midrule
				The genome-wide survival association study \citep{Huang2017}&14,406&25,849\\
				The genome-wide meta-analysis of clinically diagnosed AD and AD-by-proxy \citep{Jansen2019}&71,880&383,378\\
				The genome-wide meta-analysis of clinically diagnosed AD \citep{Kunkle2019}&21,982& 41,944\\
				The genome-wide meta-analysis by \citet{Schwartzentruber2021}&75,024&397,844\\
				In-house genome-wide associations study \citep{Belloy2022}&15,209&14,452\\
				Whole-exome sequencing analyses of data from ADSP by \citet{Bis2020}&5,740&5,096\\
				Whole-exome sequencing analyses of data from ADSP by \citet{Guen2021}&6,008&5,119\\
				In-house whole-exome sequencing analysis of ADSP&6,155&5,418\\
				In-house whole-genome sequencing analysis of the 2021 ADSP release \citep{Belloy2022a}&3,584&2,949\\
				\bottomrule
				\multicolumn{3}{l}{\footnotesize ADSP: The Alzheimer’s Disease Sequencing Project}
			\end{tabular}
			
		}
	\end{table}

\begin{figure}[t]
	\centering
	\includegraphics[width=0.8\linewidth]{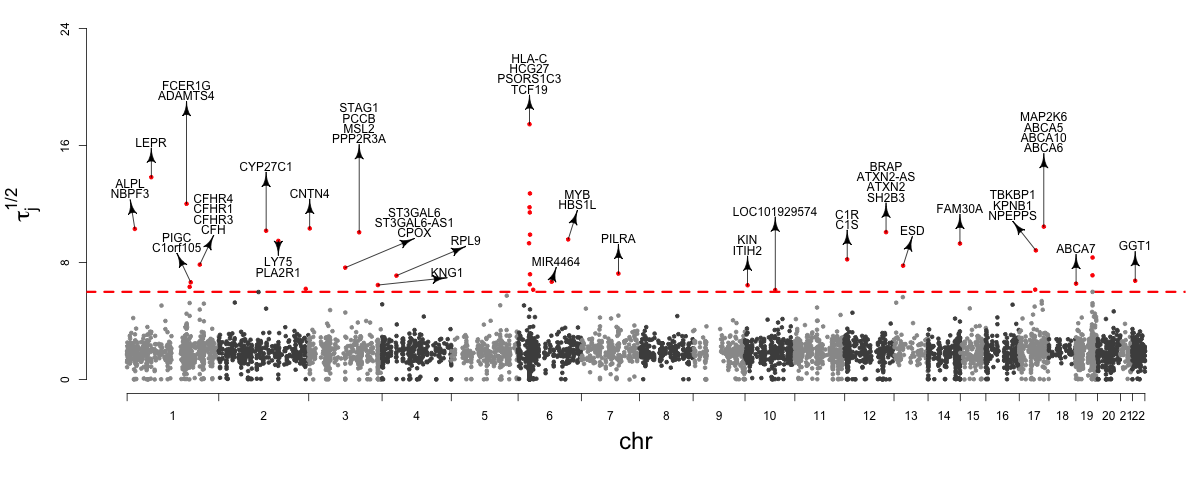}
	\includegraphics[width=0.8\linewidth]{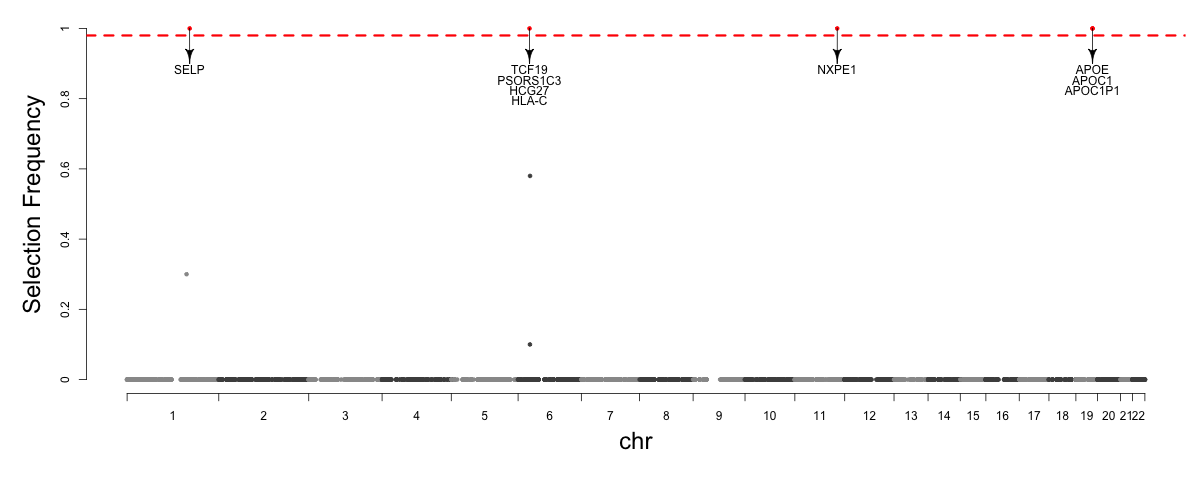}
	\caption{(Top) Manhattan plot of $\sqrt{\tau_l}$'s (truncated at 24 for better visualization) from the proposed FWER filter with GhostKnockoff with $M=19$. (Three variant groups are omitted with $\sqrt{\tau_l}$'s greatly larger than $24$, including group \texttt{19:45416178}, group \texttt{11:114408449} and group \texttt{19:45425178}). The selected variant groups under the target FWER level $\alpha=0.05$ are highlighted in red with information of close genes. (Bottom) Manhattan plot of selection frequency from derandomized knockoffs with $M_{deran}=19$. The selected variant groups under the target FWER level $\alpha=0.05$ are highlighted in red with information of close genes.}
	\label{fig:manhattonghostgrp}
\end{figure}
	
	Corresponding to all $7,963$ variants, we first estimate their covariance matrix $\boldsymbol{\Sigma}$ using genotypes from the UK Biobank data. Although we can apply the proposed FWER filter with GhostKnockoff on $Z$-scores of individual variants, directly doing so is likely to miss most of AD-associated variants. The reason is that a null variant strongly correlated to a nonnull variant is highly possible to have a large $\tau_j$ and a nonzero $\kappa_j$, making steps 5-8 of Algorithm \ref{FWER_filter} terminate at a small $T$. To circumvent such an obstacle, we apply the hierarchical clustering algorithm to build variant groups and aggregate effects of variants. Specifically, we use $(1-|\widehat{\text{cor}}(X_j,X_l)|)_{7963\times 7963}$ as the input distance matrix of variants, the ``single linkage" criterion and the cutoff distance $0.25$, leading to $5,886$ variant groups in total such that pairwise correlations between variants of different groups are within the interval $[-0.75,0.75]$.  Based on such a group structure, we aggregate the effect of all variants in the $l$-th group ($l=1,\ldots,5886$) by computing the chi-square test statistic as
	$$\chi_l=\textbf{Z}^\dT_{l\text{-th group}}\boldsymbol{\Sigma}^{-1}_{l\text{-th group}}\textbf{Z}_{l\text{-th group}},$$
	where $\textbf{Z}_{l\text{-th group}}$ (or $\boldsymbol{\Sigma}_{l\text{-th group}}$) is the $Z$-score vector  (or the covariance matrix) corresponding to variants in the $l$-th group. Analogous to $Z_j^2$ which approximates $n$ times the $R$-squared of the linear model that regresses $Y$ on the $j$-th variant, $\chi_l$ approximates $n$ times the $R$-squared of the linear model that regresses $Y$ on variants in the $l$-th group. Under the definition that groups with at least one nonnull variants are nonnull, we apply  the proposed FWER filter with GhostKnockoff to estimate nonnull groups as follows. 
	 \begin{itemize}
	 	\item[$\diamond$] Generate knockoff copies of $Z$-scores, $\widetilde{\textbf{Z}}^{1:M}$, via Algorithm  \ref{alg:eff}. Here, $\textbf{D}=\text{diag}(\textbf{S}_1,\ldots,\textbf{S}_{5886})$ is a block diagonal matrix with respect to variant groups obtained by solving the optimization problem,
	 	$$\min \sum_{l=1}^{5886}p_l^{-2}\|\textbf{S}_l-\boldsymbol{\Sigma}_l\|_{1,1}\quad \quad\text { s.t. }\begin{cases}
	 		\frac{M+1}{M} \mathbf{\Sigma-D} \succeq 0, \\
	 		\textbf{S}_l=\boldsymbol{\Lambda}_l+\sum_{k=1}^{p_l}\delta_{lk} \textbf{v}_{lk}\textbf{v}_{lk}^\dT,\quad 1 \leq l \leq 5886,\end{cases}$$
	 	where $\|\cdot\|_{1,1}$ is the $L_{1,1}$ matrix norm and for $l=1,\ldots,5886$, \begin{itemize}
	 		\item[-]  $p_l$ and $\boldsymbol{\Sigma}_l$ are the size and the correlation matrix of the $l$-th variant group respectively,
	 		 \item[-] $\boldsymbol{\Lambda}_l\succeq 0$ is a diagonal matrix,
	 		 \item[-] $\textbf{v}_{l1},\ldots,\textbf{v}_{lp_l}$ are eigenvectors of $\boldsymbol{\Sigma}_l$,
	 		 \item[-] $\delta_{l1},\ldots,\delta_{lp_l}$ are nonnegative;
	 	\end{itemize}
	\item[$\diamond$]  Compute knockoff copies of chi-square test statistics $\{\chi_l^m|l=1,\ldots,5886;m=1,\ldots,M\}$ via
	 	$$\chi^m_l=(\widetilde{\textbf{Z}}^m_{l\text{-th group}})^\dT\boldsymbol{\Sigma}^{-1}_{l\text{-th group}}\widetilde{\textbf{Z}}^m_{l\text{-th group}};$$
	 	\item[$\diamond$]  Compute the knockoff statistic of the $l$-th variant group as
	 	\begin{equation}
	 		\kappa_l=\underset{0 \leq m \leq M}{\arg \max } \text{ } (\chi_l^m)^2,\quad \tau_l=(\chi_l^{\kappa_l})^2-\underset{m\neq \kappa_l}{\operatorname{median }} (\chi_l^m)^2,\quad (l=1,\ldots,5886).
	 	\end{equation}
	 	\item[$\diamond$]  Apply Algorithm \ref{FWER_filter} on $\{(\kappa_l,\tau_l)|l=1,\ldots,5886\}$ to obtain the rejection set of variant groups.
	 \end{itemize}
 	Analogously, derandomized knockoffs \citep{zhimei} are also implemented on SDP-constructed group knockoff copies with $M_{deran}=50$ and $\eta=0.99$ using codes of \citet{zhimei}.
 	
 	With the number of knockoff copies $M=19$ and the target FWER level $\alpha=0.05$, the proposed FWER filter with GhostKnockoff manages to reject 42 variant groups as shown in top panel of Figure \ref{fig:manhattonghostgrp} and Table \ref{Real_table}. Similar to the literature, several groups (groups \texttt{19:45416178}, \texttt{19:45425178}, \texttt{19:45505803} and \texttt{19:45322744}) are rejected in the APOE/APOC region with the strongest association to AD.  In addition, the proposed FWER filter with GhostKnockoff can detect nonnull groups within 1Mb distance to genes ``ADAMTS4" (groups \texttt{1:161152778} and \texttt{1:169529132}), ``FAM20B" (group \texttt{1:172412995}), ``BIN1" (group \texttt{2:127882182}), ``HLA-DQA1" (groups \texttt{6:31195793}, \texttt{6:32202086}, \texttt{6:31229796}, \texttt{6:31900657} and \texttt{6:32623367}) and ``ABCA7" (group \texttt{19:1051137}), which are also reported in the work of \citet{he2022ghostknockoff}. In contrast, only 5 groups are rejected by derandomized knockoffs ($M_{deran}=50$ and $\eta=0.98$). As shown in bottom panel of Figure \ref{fig:manhattonghostgrp}, with only 8 variant groups have nonzero selection frequency and 5 variants have selection frequency $1$, derandomized knockoffs can not detect any variant group near genes ``BIN1" and  ``HLA-DQA1", which are of strong significance in  \citet{he2022ghostknockoff}. This is mainly due to the bounded power of the procedure of \citet{janson} underlying the derandomized procedure as shown in Section \ref{section4}, suggesting the advantage of the proposed FWER filter with GhostKnockoff.

	\renewcommand{\arraystretch}{0.7}
		\begin{table}
		\centering
		\caption{Details of variant groups rejected by the proposed FWER filter with GhostKnockoff ($M=19$) under the target FWER level $\alpha=0.05$. Groups that are also rejected by derandomized knockoffs ($M_{deran}=50$ and $\eta=0.98$) are highlighted in bold. Here, groups are named in terms of "a:b" if its leading variant is at position b of chromosome a.}\label{Real_table}
		\resizebox{\columnwidth}{!}{\scriptsize
			\begin{tabular}{llp{0.4\textwidth}p{0.5\textwidth}}
				\toprule
				Group&Chromosome&Positions of variants&Close genes\\
				\midrule\midrule
				\textbf{19:45416178}&\textbf{19}&\textbf{45411941, 45413576, 45415713, 45416178, 45416741, 45422160}&\textbf{APOE, APOC1}\\
				\textbf{11:114408449}&\textbf{11}&\textbf{114405806, 114407012, 114407750, 114408449}&\textbf{NXPE1}\\
				\textbf{19:45425178}&\textbf{19}&\textbf{45412079, 45425178, 45426792}&\textbf{APOE, APOC1P1}\\
				\textbf{6:31195793}&\textbf{6}&\textbf{31121945, 31142265, 31149520, 31154633, 31195793, 31195996, 31197074, 31200816, 31206868, 31234693}&\textbf{TCF19, PSORS1C3, HCG27, HLA-C}\\
				1:66092570&1&66092570, 66109445&LEPR\\
				6:32202086&6&31914935, 32109165, 32202086, 32278635, 32416366, 32523008, 32560266&CFB, PRRT1, TSBP1-AS1, TSBP1, HLA-DRA, HLA-DRB6, HLA-DRB1\\
				1:161152778&1&161152778, 161186313&ADAMTS4, FCER1G\\
				6:31229796&6&31229796, 31239114&HLA-C\\
				6:31900657&6&31887259, 31891491, 31893257, 31900657, 31903804, 31909340, 31919830, 32156895&C2, C2-AS1, CFB, NELFE, GPSM3\\
				17:67276383&17&67081278, 67134806, 67136325, 67249711, 67276383, 67424508&ABCA6, ABCA10, ABCA5, MAP2K6\\
				3:2713965&3&2713965, 2716366&CNTN4\\
				1:21817085&1&21775943, 21806447, 21817085, 21820042, 21821897, 21822699&NBPF3, ALPL\\
				2:127882182&2&127882182&CYP27C1\\
				12:111932800&12&111865049, 111884608, 111904371, 111907431, 111932800, 112007756, 112059557&SH2B3, ATXN2, ATXN2-AS, BRAP\\
				3:135998453&3&135798658, 135800409, 135804550, 135833005, 135846911, 135925191, 135926622, 135932359, 135950921, 135965888, 135998453, 136003159, 136020541, 136027145, 136095035, 136162621&PPP2R3A, MSL2, PCCB, STAG1\\
				6:32623367&6&32590735, 32591310, 32615457, 32623367, 32626451, 32626537, 32632770&HLA-DQA1, HLA-DQB1, HLA-DQB1-AS1\\
				6:135427817&6&135402339, 135418632, 135418916, 135421176, 135422296, 135423209, 135426573, 135427159, 135427817, 135428537, 135432552, 135435501&HBS1L, MYB\\
				2:160821211&2&160718332, 160734382, 160760972, 160778946, 160799625, 160821211&LY75, PLA2R1\\
				6:29821340&6&29821340, 29822432, 29908415, 29921619, 29923136&HCP5B, HLA-A\\
				14:106387308&14&106358616, 106363591, 106369865, 106371016, 106383775, 106384722, 106387308, 106392575&FAM30A\\
				17:45766771&17&45571676, 45735706, 45737275, 45763073, 45766771&NPEPPS, KPNB1, TBKBP1\\
				19:45505803&19&45445517, 45505803, 45507542, 45522289&APOC4-APOC2, RELB\\
				12:7178019&12&7170336, 7171338, 7171507, 7172665, 7175872, 7176204, 7176978, 7178019, 7179822, 7181948&C1S, C1R\\
				1:196822368&1&196679682, 196681001, 196698945, 196721931, 196725939, 196815711, 196819479, 196821120, 196821380, 196822368, 196825287&CFH, CFHR3, CFHR1, CFHR4\\
				13:47351403&13&47351403, 47368296&ESD\\
				3:98406933&3&98359663, 98383562, 98406794, 98406933, 98416900, 98417481, 98428155, 98429219, 98431986, 98432559, 98443648, 98453951&CPOX, ST3GAL6-AS1, ST3GAL6\\
				7:99971834&7&99971834&PILRA\\
				6:32441696&6&32441696, 32584926&HLA-DRA, HLA-DQA1\\
				19:45322744&19&45322744&BCAM\\
				4:39446337&4&39446337, 39447604&RPL9\\
				22:24997309&22&24997309, 25002323&GGT1\\
				6:90935383&6&90935383, 90960875&MIR4464\\
				1:172412995&1&172400009, 172400860, 172412995, 172421744, 172425529&C1orf105, PIGC\\
				19:1051137&19&1051137&ABCA7\\
				6:31878108&6&31878108, 31878581&C2\\
				3:186449122&3&186445052, 186449122, 186450863&KNG1\\
				10:7769806&10&7749598, 7756187, 7764233, 7769806, 7770716, 7772035, 7774358, 7774728&ITIH2, KIN\\
				\textbf{1:169529132}&\textbf{1}&\textbf{169529132, 169549811, 169562904}&\textbf{SELP}\\
				2:234665983&2&234587848, 234664586, 234665782, 234665983, 234668570, 234673309&UGT1A7, LOC100286922, UGT1A1\\
				17:44200015&17&43849415, 43897026, 44200015&CRHR1, MAPT-AS1, KANSL1-AS1\\
				6:41129207&6&41129207&TREM2\\
				10:82267945&10&82267945, 82284512&LOC101929574\\
				\bottomrule
			\end{tabular}
		}
	\end{table}

	\section{Discussions}\label{section6}

	In this article, we develop a novel filter to simultaneously test multiple hypotheses of conditional independence with FWER control using GhostKnockoff. Under the Gaussian assumption of features, we follow the procedure of \citet{he2022ghostknockoff} to directly generate $M$ knockoff copies of summary statistics without using any individual data and propose a FWER filter with GhostKnockoff to estimate the set of nonnull features. By selecting features whose knockoff statistics $\kappa_j$'s are zero and $\tau_j$'s are larger than the maximum value of $\tau_j$ among features with nonzero $\kappa_j$'s, our method manages to circumvent the power bound in \citet{janson} when the target FWER level is smaller than $1/2$.  Furthermore, an elaborated algorithm is proposed to significantly reduce the computational cost of knockoff copies generation from $O(M^3p^3)$ to $O(p^3)$ without any loss in FWER control and power. Through extensive simulations on synthetic and real datasets, the proposed FWER filter with GhostKnockoff exhibits valid control on FWER and superior power in detecting nonnull features over existing methods, including derandomized knockoffs and the procedure of \citet{janson}.
	
	There are several directions available for further studies. As the proposed FWER filter with GhostKnockoff is an one-shot  procedure, its inference has inevitably large variation due to the uncertainty in sampling knockoff copies $\widetilde{\textbf{Z}}^{1:M}$. Thus, it is possible to incorporate the proposed method in the derandomized procedure of \citet{zhimei}. By substituting the procedure of \citet{janson} to be repeated in derandomized knockoffs, we can simultaneously achieve inference stability, high power and computation efficiency. For example, using the proposed FWER filter with GhostKnockoff with $M=5$ and $v=1$ as the based procedure, the derandomized procedure \citep{zhimei} can control FWER under $\alpha=0.05$ with smaller $M_{deran}=10$ and $\eta=0.89$ compared to the setting in Section \ref{section4} ($M_{deran}=50$ and $\eta=0.99$). As shown in  simulations, both the proposed FWER filter with GhostKnockoff and derandomized knockoffs are bounded in power when features are correlated. Such a phenomenon is mainly due to the mismatching between $Z$-scores, which generally depict marginal correlations, and the conditional independence (\ref{indep}) we target to test. Therefore, it is worth investigating how to develop feature importance scores to directly depict conditional correlations with only summary statistics. One possible solution is to use $\mathbf{\widetilde{Z}}^{0},\dots \mathbf{\widetilde{Z}}^{M}$ and matrices $\boldsymbol{\Sigma}$ and $\textbf{D}$ to obtain the approximate lasso estimator of the linear model 
		$Y=  \boldsymbol{\beta}^\dT\mathbf{X}+\sum_{m=1}^M(\widetilde{\boldsymbol{\beta}}^{m})^\dT\widetilde{\mathbf{X}}^{m} +\epsilon$
as feature importance scores. By doing so, magnitudes of feature importance scores corresponding to nonnull features would increase as the signal amplitude $A$ grows and keep unchanged as the number of nonnull features $|\mathcal{H}_1|$ increases, getting rid of the power bound phenomenon shown in Section \ref{section4}. 
	In addition, when the number of features $p$ is large as in genome-wide analyses, it is highly possible to have some null features with nonzero $\kappa_j$'s and large $\tau_j$'s, leading to early stop in nonnull features selection and  power loss. Thus, it is of great interest to incorporate some power boosting strategies to the proposed FWER filter with GhostKnockoff, including the feature screening technique \citep{Barber2019} and incorporating side information \citep{Ren2023}.

	\section*{Acknowledgements}

	We would like to thank Emmanuel Cand\`es for helpful discussions and useful comments about an early version of the manuscript. This research was additionally supported by NIH/NIA award AG066206 (Zihuai He) and AG066515 (Zihuai He).
	
	
%
	\vspace*{-8pt}

	\bibliographystyle{apalike} 
	\bibliography{biomsample}

	\begin{appendices}
		\section{Proof of Lemma 1}\label{proof:lemma1}
			\begin{proof}[Proof of Lemma \ref{lemma1}]
			Let $\widetilde{Z}_j^{m,\star}=\widetilde{Z}_j^{\sigma_j(m)}$ for $j=1,\ldots,p$ and $m=0,\ldots,M$, we have for any $\boldsymbol{\sigma}=\{\sigma_1,\ldots,\sigma_p\}$ where $\sigma_j$ is any permutation on $\{0,1,\ldots,M\}$ if $H_j$ is true and is the identity permutation if $H_j$ is false ($j=1,\ldots,p$), 
			\begin{equation}
				\begin{cases}
					\kappa_j^\star=\underset{0 \leq m \leq M}{\arg \max } \text{ } (\widetilde{Z}_j^{m,\star})^2=\sigma_j^{-1}(\kappa_j),&j\in \mathcal{H}_0,\\
					\kappa_j^\star=\underset{0 \leq m \leq M}{\arg \max } \text{ } (\widetilde{Z}_j^{m,\star})^2=\kappa_j,&j\in \mathcal{H}_1,\\
					\tau^\star_j=(\widetilde{Z}_j^{\kappa_j})^2-\underset{m\neq \kappa_j}{\operatorname{median }} (\widetilde{Z}_j^{m,\star})^2=\tau_j,&j=1,\ldots,p.\\
				\end{cases}
			\end{equation}
			Since $Z$-scores $\mathbf{\widetilde{Z}}^{0},\dots \mathbf{\widetilde{Z}}^{M}$ possess the extended exchangeability property (\ref{exchangeablility}) with respect to $\{H_j:j=1,\ldots,p\}$, we have 
			$$\bigg[\{\kappa_j|j\in \mathcal{H}_0\}\bigg|\{\kappa_j|j\notin \mathcal{H}_0\},\{\tau_j|j=1,\ldots,p\}\bigg]{\buildrel d \over =}\bigg[\{\sigma_j^{-1}(\kappa_j)|j\in \mathcal{H}_0\}\bigg|\{\kappa_j|j\notin \mathcal{H}_0\},\{\tau_j|j=1,\ldots,p\}\bigg].$$
			Because $\sigma_j$ can be any permutation for all $j\in \mathcal{H}_0$, $\sigma_j^{-1}$ can also be any permutation. Thus, $\{\kappa_j|j\in \mathcal{H}_0\}$ are independent uniform random variables on $\{0,1, \ldots, M\}$ conditional on $\{\kappa_j|j\notin \mathcal{H}_0\}$ and $\{\tau_j|j=1,\ldots,p\}$.
		\end{proof}
		
		\section{Proof of extended exchangeability property of $Z$-scores $\mathbf{\widetilde{Z}}^{0},\mathbf{\widetilde{Z}}^{1},\dots \mathbf{\widetilde{Z}}^{M}$}\label{Proof:exchangeability}
		
		As shown in \citet{he2022ghostknockoff}, when the Gaussian model $\textbf{X}\sim \text{N}(\mathbf{0},\boldsymbol{\Sigma})$ is assumed to generate $M$-knockoffs $(\mathbf{\widetilde{X}}^{1},\dots \mathbf{\widetilde{X}}^{M})$, knockoff copies of $Z$-score $\mathbf{\widetilde{Z}}^{1},\dots \mathbf{\widetilde{Z}}^{M}$ defined by (\ref{zscore_knockoff}) satisfy (\ref{Knockoff_conditional_Z}). Thus, we have 
		$\widetilde{Z}^m_j$ is a function of $\widetilde{\mathbb{X}}^m_j=\{\widetilde{x}^m_{ij}:i=1,\ldots,n\}$ and $\mathbb{Y}=\{y_i:i=1,\ldots,n\}$
		$$\widetilde{Z}^m_j=\nu(\widetilde{\mathbb{X}}^m_j,\mathbb{Y})$$ for any $m=0,\ldots,M$ and $j=1,\ldots,p$ under the convention $\mathbf{\widetilde{x}}_i^{0}=\textbf{x}_i$ ($i=1,\ldots,n$).
		
		For any $\boldsymbol{\sigma}=\{\sigma_1,\ldots,\sigma_p\}$ where 
		\begin{equation}\label{cond1}
			\begin{cases}
				\sigma_j\text{ is any permutation on }\{0,1,\ldots,M\},&\text{if }H_j\text{ is true (or }j\in\mathcal{H}_0\text{),}\\
				\sigma_j\text{ is the identity permutation, }&\text{if }H_j\text{ is false (or }j\in\mathcal{H}_1\text{),}\\
			\end{cases}
		\end{equation}we have 
		\begin{equation}\label{cond2}
			\begin{aligned}[b]
			\mathbf{\widetilde{Z}}^{0:M}_{swap(\boldsymbol{\sigma})}&=(\widetilde{Z}^{\sigma_1(0)}_1,\ldots,\widetilde{Z}^{\sigma_p(0)}_p,\ldots,\widetilde{Z}^{\sigma_1(M)}_1,\ldots,\widetilde{Z}^{\sigma_p(M)}_p)^\dT\\
			&=(\nu(\widetilde{\mathbb{X}}^{\sigma_1(0)}_1,\mathbb{Y}),\ldots,\nu(\widetilde{\mathbb{X}}^{\sigma_p(0)}_p,\mathbb{Y}),\ldots,\nu(\widetilde{\mathbb{X}}^{\sigma_1(M)}_1,\mathbb{Y}),\ldots,\nu(\widetilde{\mathbb{X}}^{\sigma_p(M)}_p,\mathbb{Y}))^\dT.
		\end{aligned}
		\end{equation}
		
		By the conditional independence property of multiple knockoffs \citep{james} that 
		 $$(\mathbf{\widetilde{X}}^{1},\dots \mathbf{\widetilde{X}}^{M})\perp Y|\mathbf{\widetilde{X}}^{0}$$
		 and hypotheses (\ref{indep}), we have $(\mathbf{\widetilde{X}}^{0}_{\mathcal{H}_0},\mathbf{\widetilde{X}}^{1},\dots \mathbf{\widetilde{X}}^{M})\perp Y|\mathbf{\widetilde{X}}^{0}_{\mathcal{H}_1}$ where $\mathbf{\widetilde{X}}^{0}_{\mathcal{H}_0}$ (or $\mathbf{\widetilde{X}}^{0}_{\mathcal{H}_1}$) is the subvector of $\mathbf{\widetilde{X}}^{0}$ corresponding to $\mathcal{H}_0$ (or $\mathcal{H}_1$). Thus, the joint distribution of $(\mathbf{\widetilde{X}}^{0},\mathbf{\widetilde{X}}^{1},\dots \mathbf{\widetilde{X}}^{M},Y)$ satisfies
		 $$\begin{aligned}[b]
		 	F(\mathbf{\widetilde{X}}^{0:M},Y)&=F(\mathbf{\widetilde{X}}^{0}_{\mathcal{H}_1})F(\mathbf{\widetilde{X}}^{0}_{\mathcal{H}_0},\mathbf{\widetilde{X}}^{1},\dots \mathbf{\widetilde{X}}^{M}|\mathbf{\widetilde{X}}^{0}_{\mathcal{H}_1})F(Y|\mathbf{\widetilde{X}}^{0}_{\mathcal{H}_1})\\
		 	&=F(\mathbf{\widetilde{X}}^{0:M})F(Y|\mathbf{\widetilde{X}}^{0}_{\mathcal{H}_1})
		 \end{aligned}$$
		 
		 By the exchangeability property of multiple knockoffs \citep{james} that $$\mathbf{\widetilde{X}}^{0:M}_{swap(\boldsymbol{\sigma})}=(\widetilde{X}^{\sigma_1(0)}_1,\ldots,\widetilde{X}^{\sigma_p(0)}_p,\ldots,\widetilde{X}^{\sigma_1(M)}_1,\ldots,\widetilde{X}^{\sigma_p(M)}_p)^\dT\,{\buildrel d \over =}\,\mathbf{\widetilde{X}}^{0:M},$$
		 we have for any $\boldsymbol{\sigma}$ satisfies (\ref{cond1}), 
		 \begin{equation}\label{cond3}
		 	\begin{aligned}[b]
		 	F(\mathbf{\widetilde{X}}_{swap(\boldsymbol{\sigma})}^{0:M},Y)
		 	&=F(\mathbf{\widetilde{X}}_{swap(\boldsymbol{\sigma})}^{0:M})F(Y|\mathbf{\widetilde{X}}^{0}_{\mathcal{H}_1})\\
		 	&=F(\mathbf{\widetilde{X}}^{0:M})F(Y|\mathbf{\widetilde{X}}^{0}_{\mathcal{H}_1})\\
		 	&=F(\mathbf{\widetilde{X}}^{0:M},Y).
		 \end{aligned}
		 \end{equation}
		
		In conclusion, we have for any $\boldsymbol{\sigma}$ satisfies (\ref{cond1}), by (\ref{cond2}) and (\ref{cond3}), 
		$$F(\mathbf{\widetilde{Z}}_{swap(\boldsymbol{\sigma})}^{0:M}){\buildrel d \over =}F(\mathbf{\widetilde{Z}}^{0:M}),$$
		and thus $\mathbf{\widetilde{Z}}^{0:M}$ possess extended exchangeability property (\ref{exchangeablility}).
	
		\newpage
	
		\section{Additional Simulation Results}
		
		\subsection{Correlation Structure}\label{Supp41}
		
		\begin{figure}[h]
			\begin{subfigure}[b]{0.33\textwidth}
				\centering
				\includegraphics[width=\textwidth]{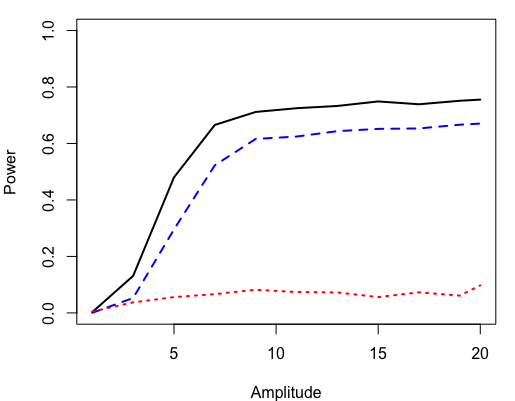}
				\caption{Correlation strength $\rho=0.25$.}
			\end{subfigure}
			\begin{subfigure}[b]{0.33\textwidth}
				\centering
				\includegraphics[width=\textwidth]{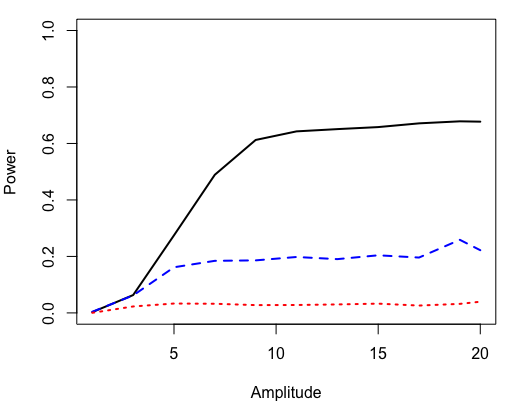}
				\caption{Correlation strength $\rho=0.5$.}
			\end{subfigure}
			\begin{subfigure}[b]{0.33\textwidth}
				\centering
				\includegraphics[width=\textwidth]{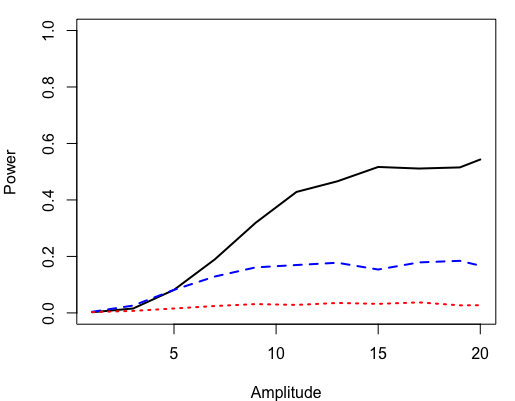}
				\caption{Correlation strength $\rho=0.75$.}
			\end{subfigure}
			\caption{Empirical power over $500$ simulated datasets of the proposed FWER filter with GhostKnockoff (black solid line), derandomized knockoffs (blue dashed line) and the procedure of \citet{janson} (red dotted line) with respect to different signal amplitudes $A$ and different correlation strengths under compound symmetry correlation structure, sample size $n=500$, dimension $p=100$, number of nonnull features $|\mathcal{H}_1|=5$, the target FWER level $\alpha=0.05$.  }
		\end{figure}
	
		\begin{figure}[h]
			\begin{subfigure}[b]{0.33\textwidth}
				\centering
				\includegraphics[width=\textwidth]{cs_100_2.png}
				\caption{Dimension $p=100$.}
			\end{subfigure}
			\begin{subfigure}[b]{0.33\textwidth}
				\centering
				\includegraphics[width=\textwidth]{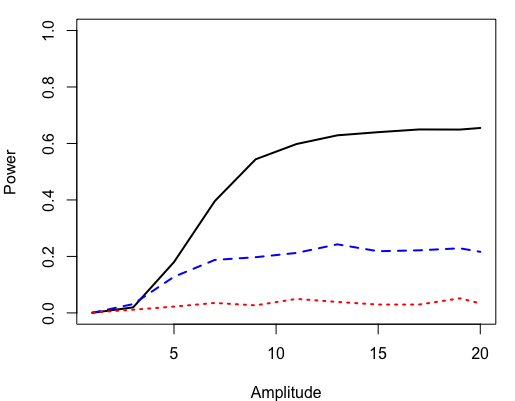}
				\caption{Dimension $p=500$.}
			\end{subfigure}
			\begin{subfigure}[b]{0.33\textwidth}
				\centering
				\includegraphics[width=\textwidth]{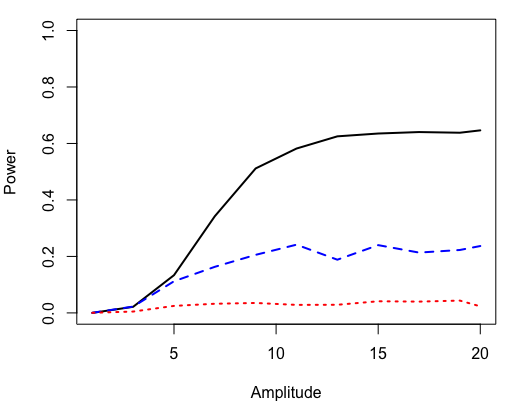}
				\caption{Dimension $p=1000$.}
			\end{subfigure}
			\caption{
				Empirical power over $500$ simulated datasets of the proposed FWER filter with GhostKnockoff (black solid line), derandomized knockoffs (blue dashed line) and the procedure of \citet{janson} (red dotted line) with respect to different signal amplitudes $A$ and different dimensions under compound symmetry correlation structure, sample size $n=500$, correlation strength $\rho=0.5$, number of nonnull features $|\mathcal{H}_1|=5$, the target FWER level $\alpha=0.05$.}
		\end{figure}

			\begin{figure}[h]
		\begin{subfigure}[b]{0.33\textwidth}
			\centering
			\includegraphics[width=\textwidth]{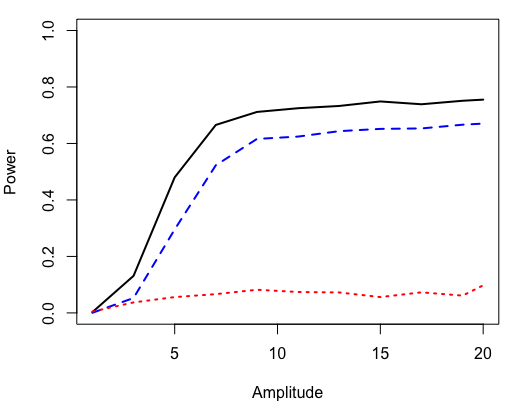}
			\caption{Correlation strength $\rho=0.25$.}
		\end{subfigure}
		\begin{subfigure}[b]{0.33\textwidth}
			\centering
			\includegraphics[width=\textwidth]{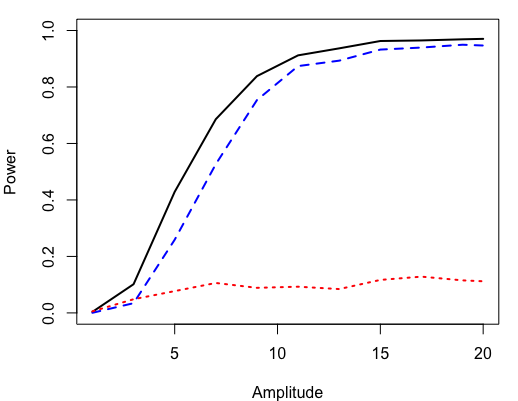}
			\caption{Correlation strength $\rho=0.5$.}
		\end{subfigure}
		\begin{subfigure}[b]{0.33\textwidth}
			\centering
			\includegraphics[width=\textwidth]{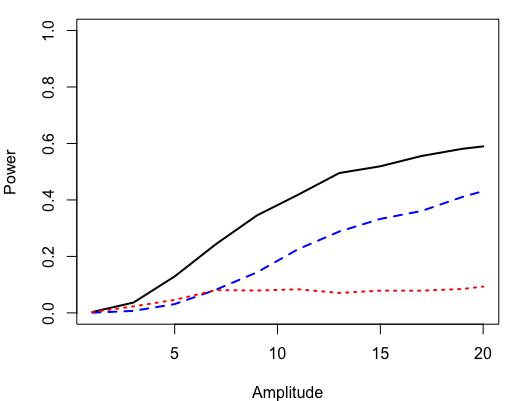}
			\caption{Correlation strength $\rho=0.75$.}
		\end{subfigure}
		\caption{Empirical power over $500$ simulated datasets of the proposed FWER filter with GhostKnockoff (black solid line), derandomized knockoffs (blue dashed line) and the procedure of \citet{janson} (red dotted line) with respect to different signal amplitudes $A$ and different correlation strengths under compound correlation structure, sample size $n=500$, dimension $p=100$, number of nonnull features $|\mathcal{H}_1|=5$, the target FWER level $\alpha=0.05$.  }
	\end{figure}
	
	\begin{figure}[h]
		\begin{subfigure}[b]{0.33\textwidth}
			\centering
			\includegraphics[width=\textwidth]{ar_100_2.png}
			\caption{Dimension $p=100$.}
		\end{subfigure}
		\begin{subfigure}[b]{0.33\textwidth}
			\centering
			\includegraphics[width=\textwidth]{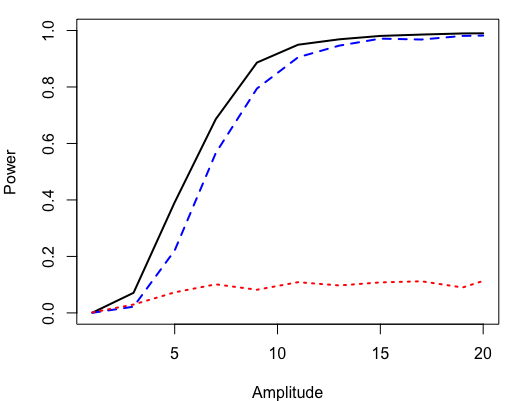}
			\caption{Dimension $p=500$.}
		\end{subfigure}
		\begin{subfigure}[b]{0.33\textwidth}
			\centering
			\includegraphics[width=\textwidth]{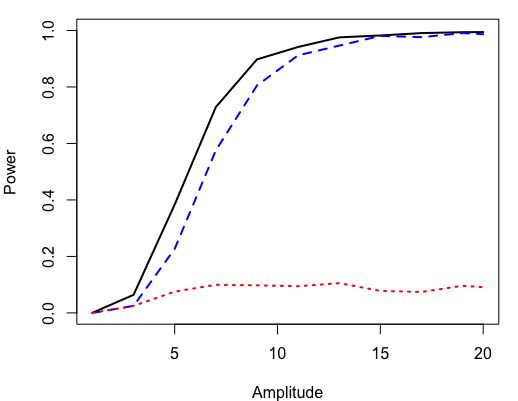}
			\caption{Dimension $p=1000$.}
		\end{subfigure}
		\caption{Empirical power over $500$ simulated datasets of the proposed FWER filter with GhostKnockoff (black solid line), derandomized knockoffs (blue dashed line) and the procedure of \citet{janson} (red dotted line) with respect to different signal amplitudes $A$ and different dimensions under AR(1) correlation structure, sample size $n=500$, correlation strength $\rho=0.5$, number of nonnull features $|\mathcal{H}_1|=5$, the target FWER level $\alpha=0.05$.}
	\end{figure}
	
	\newpage

		\subsection{Number of Nonnull Features and Sample Size}\label{Supp43}
				\begin{figure}[h]
				\begin{subfigure}[b]{0.33\textwidth}
			\centering
			\includegraphics[width=\textwidth]{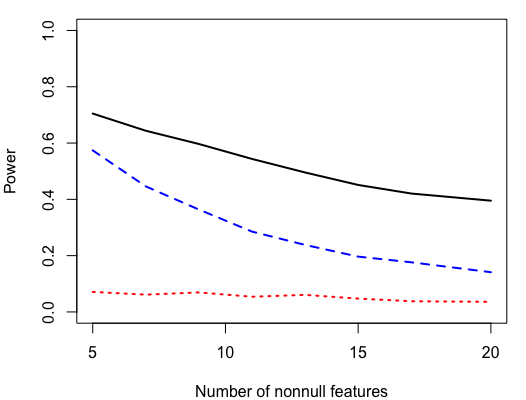}
			\caption{Correlation strength $\rho=0.25$.}
		\end{subfigure}
		\begin{subfigure}[b]{0.33\textwidth}
			\centering
			\includegraphics[width=\textwidth]{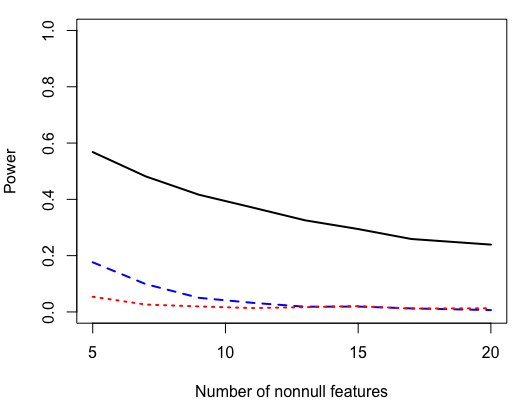}
			\caption{Correlation strength $\rho=0.5$.}
		\end{subfigure}
		\begin{subfigure}[b]{0.33\textwidth}
			\centering
			\includegraphics[width=\textwidth]{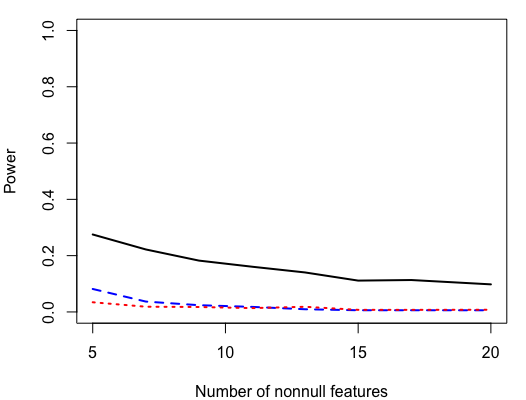}
			\caption{Correlation strength $\rho=0.75$.}
		\end{subfigure}
		\caption{Empirical power over $500$ simulated datasets of the proposed FWER filter with GhostKnockoff (black solid line), derandomized knockoffs (blue dashed line) and the procedure of \citet{janson} (red dotted line) with respect to different numbers of nonnull features $|\mathcal{H}_1|$ and different correlation strengths under compound symmetry correlation structure, sample size $n=500$, dimension $p=100$, amplitude strength $A=6$, the target FWER level $\alpha=0.05$.  }
	\end{figure}
		
				\begin{figure}[h]
	\begin{subfigure}[b]{0.33\textwidth}
		\centering
		\includegraphics[width=\textwidth]{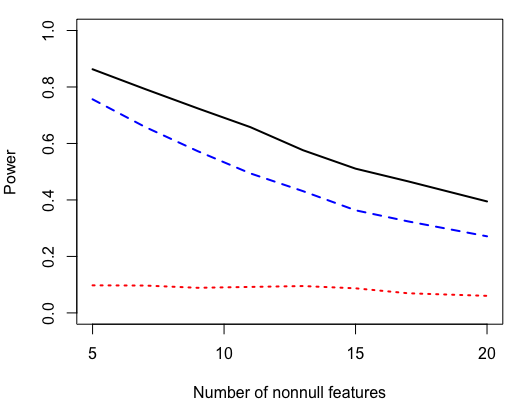}
		\caption{Correlation strength $\rho=0.25$.}
	\end{subfigure}
	\begin{subfigure}[b]{0.33\textwidth}
		\centering
		\includegraphics[width=\textwidth]{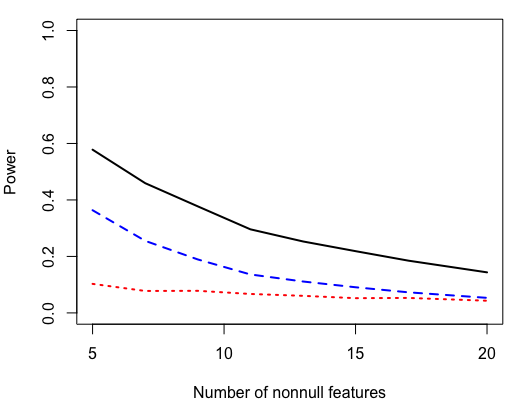}
		\caption{Correlation strength $\rho=0.5$.}
	\end{subfigure}
	\begin{subfigure}[b]{0.33\textwidth}
		\centering
		\includegraphics[width=\textwidth]{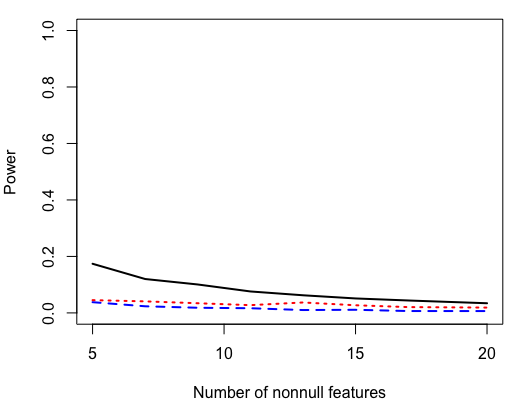}
		\caption{Correlation strength $\rho=0.75$.}
	\end{subfigure}
	\caption{Empirical power over $500$ simulated datasets of the proposed FWER filter with GhostKnockoff (black solid line), derandomized knockoffs (blue dashed line) and the procedure of \citet{janson} (red dotted line) with respect to different numbers of nonnull features $|\mathcal{H}_1|$ and different correlation strengths under AR(1) correlation structure, sample size $n=500$, dimension $p=100$, amplitude strength $A=8$, the target FWER level $\alpha=0.05$. }
\end{figure}

\begin{figure}[h]
	\begin{subfigure}[b]{0.33\textwidth}
		\centering
		\includegraphics[width=\textwidth]{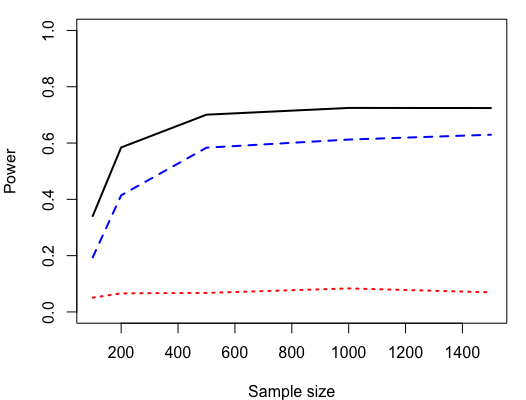}
		\caption{Correlation strength $\rho=0.25$.}
	\end{subfigure}
	\begin{subfigure}[b]{0.33\textwidth}
		\centering
		\includegraphics[width=\textwidth]{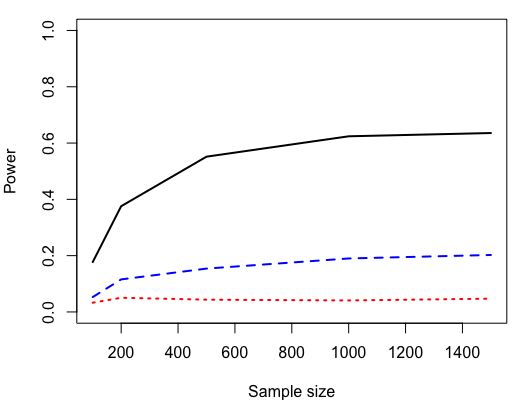}
		\caption{Correlation strength $\rho=0.5$.}
	\end{subfigure}
	\begin{subfigure}[b]{0.33\textwidth}
		\centering
		\includegraphics[width=\textwidth]{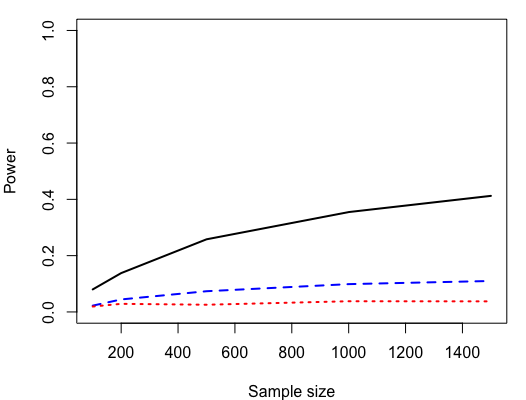}
		\caption{Correlation strength $\rho=0.75$.}
	\end{subfigure}
	\caption{Empirical power over $500$ simulated datasets of the proposed FWER filter with GhostKnockoff (black solid line), derandomized knockoffs (blue dashed line) and the procedure of \citet{janson} (red dotted line) with respect to different sample sizes $n$ and different correlation strengths under compound symmetry correlation structure, dimension $p=100$, number of nonnull features $|\mathcal{H}_1|=5$, amplitude strength $A=6$, the target FWER level $\alpha=0.05$.  }
\end{figure}

			\begin{figure}[h]
			\begin{subfigure}[b]{0.33\textwidth}
				\centering
				\includegraphics[width=\textwidth]{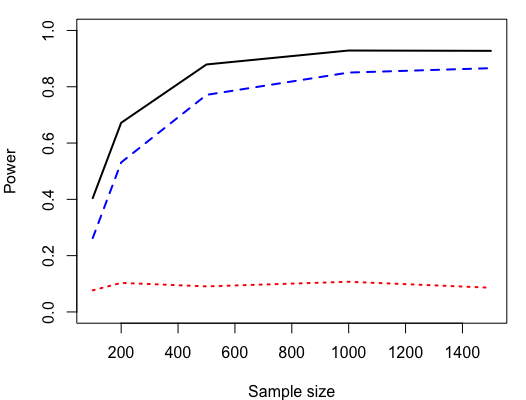}
				\caption{Correlation strength $\rho=0.25$.}
			\end{subfigure}
			\begin{subfigure}[b]{0.33\textwidth}
				\centering
				\includegraphics[width=\textwidth]{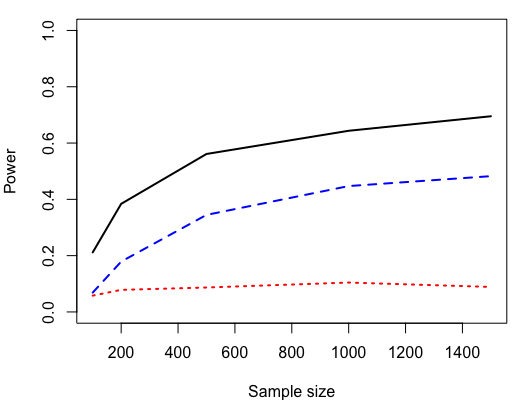}
				\caption{Correlation strength $\rho=0.5$.}
			\end{subfigure}
			\begin{subfigure}[b]{0.33\textwidth}
				\centering
				\includegraphics[width=\textwidth]{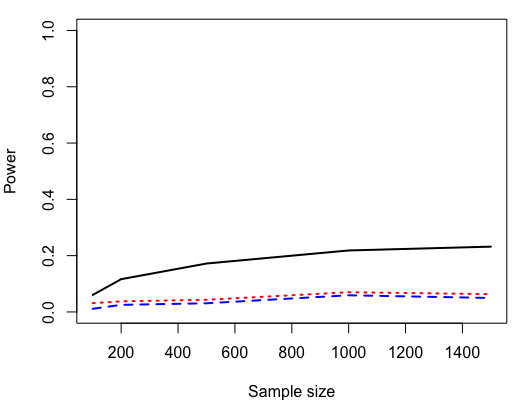}
				\caption{Correlation strength $\rho=0.75$.}
			\end{subfigure}
			\caption{Empirical power over $500$ simulated datasets of the proposed FWER filter with GhostKnockoff (black solid line), derandomized knockoffs (blue dashed line) and the procedure of \citet{janson} (red dotted line) with respect to different sample sizes $n$ and different correlation strengths under AR(1) correlation structure, dimension $p=100$, number of nonnull features $|\mathcal{H}_1|=5$, amplitude strength $A=8$, the target FWER level $\alpha=0.05$.  }
		\end{figure}

\end{appendices}

\end{document}